\documentclass[copyright,creativecommons]{eptcs} 
\providecommand{\event}{DCM 23} 
\usepackage{underscore}           

\usepackage{mathtools,mathpartir,amsmath,amssymb}
\usepackage{hyperref}
\usepackage{keyval}
\usepackage{prooftree}
\usepackage[all]{xy}
\usepackage{dsfont}
\usepackage{tikz}                  
\usepackage{enumerate} 
\usepackage{placeins}                 
\usetikzlibrary{shadows,arrows,shapes,automata,positioning,decorations.markings,fit}
\usepgflibrary{arrows.spaced}
\usepackage{wrapfig}
\usepackage{breakurl}
\usepackage[misc,geometry]{ifsym}
\usepackage[algoruled,vlined]{algorithm2e}
\SetKwProg{Fn}{Function}{}{}

\usepackage[nomargin,inline,index,%
]{fixme}

\usepackage[shortlabels]{enumitem}
\usepackage{makeidx}
\usepackage{color}
\usepackage{soul}

\usepackage{mathpartir,amsmath,amssymb
}

 \newtheorem{theorem}{Theorem}[section]
\newtheorem{definition}[theorem]{Definition}
 
 \newtheorem{lemma}[theorem]{Lemma}
 \newtheorem{example}[theorem]{Example}
 
 \newtheorem{fact}[theorem]{Fact}
 \newtheorem{remark}[theorem]{Remark}
\newenvironment{proof}{{\em Proof.}}{\hfill$\square$\medskip}

\newcommand{\refToDef}[1]{Definition~\ref{#1}}

\newcommand{\inact}{\ensuremath{\mathbf{0}}}
\newcommand{\coDefGr}{::=_\rho}
\newcommand{\oup}[5]{#1!\set{#4_#2.#5_#2}_{#2\in#3}}
\newcommand{\inp}[5]{#1?\set{#4_#2.#5_#2}_{#2\in#3}}

\newcommand{\PP}{\ensuremath{P}}
\newcommand{\PS}{\ensuremath{S}}
\newcommand{\PU}{\ensuremath{U}}
\newcommand{\Q}{\ensuremath{Q}}
\newcommand{\PR}{\ensuremath{R}}
\newcommand{\R}{\ensuremath{R}}
\newcommand{\pp}{{\sf p}}
\newcommand{\q}{{\sf q}}
\newcommand{\pr}{{\sf r}}
\newcommand{\ps}{{\sf s}}
\newcommand{\pu}{{\sf u}}

\newcommand{\la}{\lambda}
\newcommand{\G}{\ensuremath{{\sf G}}}
\newcommand{\End}{\sf{End}}
\newcommand{\agtO}[6]{#1#2!\{#5_#3.#6_#3\}_{#3\in#4}}
\newcommand{\agtI}[6]{#1#2?#5.#6}

  \newcommand{\rn}[1]{{[\textsc{#1}]}}

 \newcommand{\pP}[2] {#1[\![\,#2\,]\!]}
 \newcommand{\Nt}{{\mathbb{N}}} 
  \newcommand{\parN}{\mathrel{\|}}
  \newcommand{\plays}[1]{\ensuremath{{\sf Prt}(#1)}}
  \newcommand{\Plays}[1]{\ensuremath{{\sf Plays}(#1)}}
  \newcommand{\set}[1]{\{#1\}}
   \newcommand{\Set}[1]{\set{#1}}
  
  \newcommand{\Msg}{\mathcal{M}} 
  \newcommand{\mq}[3]{\langle#1,#2,#3\rangle}
  \newcommand{\parG}{\mathrel{\|}}
  \newcommand{\addMsg}[2]{#1\cdot #2}

  \newcommand{\val}{v}
  \newcommand{\confAs}[2]{#1\parN#2}
  \newcommand{\stackred}[1]{\xrightarrow{#1}}
  \newcommand{\CommAs}[3]{#1#3!#2}
  \newcommand{\CommAsI}[3]{#1#3?#2}
  
  \newcommand{\asCom}{\beta}
  
  \newcommand{\comseqA}{ \tau }
   \newcommand{\mypath}{\comseqA}
  
   \newcommand{\ee}{\epsilon}
\newcommand{\concat}[2]{\ensuremath{#1\,{\cdot}\,#2}}

\newcommand{\play}[1]{\ensuremath{{\sf play}(#1)}}
\newcommand{\Seq}[2]{#1.#2}
\newcommand{\NamedRule}[5][]{ \Infer[#1]{#2}{ #3 }{#4}{#5} }
\newcommand {\Infer} [5] [] {
  \inferrule*[%
    left={\textsc{#2}},%
    right={$\begin{array}{l} {#5} \end{array}$}, 
    vcenter,%
    #1
  ]%
  {#3}{#4}}
  \newcommand{\co}{\beta}
   
  \newcommand{\ipth}{\xi}

\newcommand{\IPaths}[1]{{\sf Paths}(#1)}
\newcommand{\weight}{\ensuremath{{\sf depth}}}
\newcommand{\agtoneO}[4]{\Seq{#1#2!#3}{#4}}
\newcommand{\agtIS}[3]{#1#2?#3}
\newcommand{\agtSOS}[3]{#1#2!#3}
\newcommand{\Hset}{\mathcal{H}}

\newcommand{\wgs}[2]{\ensuremath{{\sf weight}(#2,#1)}}

\newcommand{\pair}[2]{\ensuremath{(#1,#2)}}

 \newcommand{\tyng}[3]{\vdash_{#1} #3:#2}
 \newcommand{\pas}{\mathcal{P}}

 \newcommand{\msg}[1]{\text{\footnotesize \sc #1}}

\newcommand{\outpone}[3]{#1!\Seq{#2}{#3}}

\newcommand{\M}{\lambda}
 
 \newcommand{\Cline}[1]{\centerline{$ #1 $}} 
 \newcommand{\ClineL}[1]{\vspace{2.2mm}\centerline{$ #1 $}\vspace{2.2mm}} 

\begin{document}

\title{Partial Typing for Asynchronous Multiparty Sessions}
\author{Franco Barbanera
\thanks{
Partially supported by 
Project ``National Center for HPC, Big Data e Quantum Computing",  Programma M4C2, Investimento 1.3 – Next Generation EU.
}
\institute{
Dipartimento di Matematica e Informatica,
Universit\`a di  Catania, Catania, Italy}
\email{franco.barbanera@unict.it}
\and
Mariangiola Dezani-Ciancaglini
\qquad \qquad\qquad 
Ugo de'Liguoro
\thanks{Partially supported by Project INDAM-GNCS ``Fondamenti di Informatica e Sistemi Informatici".
}
\institute{
Dipartimento di Informatica,
Universit\`a di Torino, Torino, Italy}
\email{\{dezani,deligu\}@di.unito.it}
}

\def\titlerunning{Partial Typing for Asynchronous  Multiparty Sessions} 

\def\authorrunning{ Barbanera \& Dezani-Ciancaglini \& de'Liguoro
}

\maketitle

\begin{abstract}
Formal verification methods for concurrent systems cannot always be
scaled-down or tailored in order to be applied on specific subsystems.
We address such an issue in a MultiParty Session Types setting by devising
a {\em partial} type assignment system for multiparty sessions (i.e. sets of concurrent participants) with asynchronous communications. 
Sessions are possibly typed by ``asynchronous global types'' describing the overall behaviour of specific subsets of participants only (from which the word  
``partial''). Typability is proven to ensure  that sessions enjoy  the partial versions of the
well-known properties of lock- and orphan-message-freedom. 
 \end{abstract}
 

  \section{Introduction}\label{intro}

When validating/verifying  distributed and concurrent systems, it is often natural to identify different subsystems
for which the properties we have to take into account are not those required for the whole system, if any.
The system of a social media, for instance, is made of users and services the former are provided with.
The users are the main concern of the social media, which hence tend to  ensure 
to the user subsystem
properties which cannot be (or need not to be)  ensured  
to the services. 
This particularly applies in case services are  managed by a second 
party not under direct control of the social media.
{\em Lock-freedom} is a relevant specimen of such properties.
It ensures that no lock is ever reached in the evolution of a system. 
A lock being a system’s reachable configuration where a still active participant is
forever prevented to perform any action in any possible continuation of the system\footnote{Actually several slightly different property are 
present in the literature under the name ``lock-freedom''.}.
In particular, such a configuration is called a $\pp$-lock in case the stuck participant  is  $\pp$. 
A social media would hence be focused on $\pp$-lock freedom for each $\pp\in\pas$, where $\pas$ is the set of
 users in the current example. As far as the users cannot get into a lock, the services can behave as they like best.
The social media can also be interested in that, in case of an asynchronous model of communication, the
messages exchanged among the users are eventually received. This is a {\em partial} version of the
property referred to in the literature as {\em orphan-message freedom}.
An investigation on verification of partial properties was  carried on in \cite{BD23} in the 
setting of MultiParty Session Types (MPST for short),  in particular in a {\em bottom-up} MPTS setting.
 Unlike formalisms using the notion of {\em projections}, the formalism in \cite{BD23} enables to exploit an
approach to the development and verification of distributed/concurrent system where systems (formalised here through the notion of ``network'', a parallel composition of  named  processes) are first developed
and then  subsequently proved sound with respect a specific overall description of the system's behaviour
by checking the network against a {\em global type}.  
The  MPST  type  system of \cite{BD23} 
derives judgements of the   shape

\ClineL{
\tyng{\pas}{\G}{\Nt}
}
\noindent
where $\pas$  is a  set of participants, $\Nt$ is a network and $\G$ is a global type. 
The typing is {\em partial} since some communications between participants in $\pas$ do not appear in the global type. Typing  $\Nt$ with $\G$
 does ensure that $(a)$ the communications of the participants in $\Nt$ not belonging 
to  $\pas$ comply with the interaction scenario represented by $\G$ and
$(b)$ $\Nt$ is $\pp$-lock-free for each  $\pp \not\in \pas$. 

In the present paper we push further the investigation  of~\cite{BD23}  by  treating 
an asynchronous
model of communication, instead of a synchronous one. 
Besides, we take into account also the partial version of the property of {\em orphan-message freedom}.
The calculus, the global types and the type system  we use are inspired by~\cite{CDG21,CDG22,DGD22}.

\bigskip

{\em Contributions and structure of the paper.} In Section~\ref{ms} we recall  from~\cite{CDG21}
the asynchronous calculus of  multiparty sessions. 
  Also, we adapt from~\cite{BD23} the notion of $\pas$-lock-freedom
   (the absence of locks is ensured here to the participants in $\pas$) 
   and introduce 
the novel notion of $\pas$-orphan-message freedom. 
An example is  given  to clarify the various notions and results.   
Section~\ref{ts} is devoted to the presentation of  (asynchronous)  global types 
 from~\cite{CDG21} and the introduction of our 
 ``partial'' type system, assigning global types to multiparty sessions, where some communications can be ignored.  
The relevant properties of  partially  typable sessions are proved in
Section~\ref{pr}. In particular Subject Reduction, Session Fidelity,
 $\pas$-lock-freedom and  $\pas$-orphan-message-freedom.
A section summing up our results, 
discussing related works and possible directions for future work concludes  the paper.

 \section{Multiparty Sessions}\label{ms}

The calculus of multiparty sessions, as well  as  global types, used in  the present paper
are inspired  by~\cite{CDG21}. 
The simplicity of the calculus with respect to the original  MPST calculus~\cite{HYC08} and of many of the subsequent ones, as well as the lack of explicit  channels,  enables  us to focus on our main concerns.  Besides,  it allows 
for 
a clear explanation of the type system we will introduce in the next section. 
All this has however the cost of preventing the representation of session interleaving and delegation.

We use the following base sets and notation: \emph{labels}, ranged
over by $\lambda,\lambda',\dots$; \emph{session participants}, ranged over
by $\pp,\q,\pr, \ps, \pu,\ldots$; \emph{processes}, ranged over by $\PP,\Q,\PR,\PS,\PU,\dots$;
\emph{networks}, ranged over by $\Nt,\Nt',\dots$;
\emph{queues}, ranged over by $\Msg,\Msg',\dots$;
\emph{integers}, ranged over by $i, j,l,h,k,
\dots$;   {\em  (finite)  integer sets}, ranged over by $I,J,L, H, K,\dots$.

  \begin{definition}[Processes]\label{p} 
  {\em Processes} are defined by:
  
\ClineL{\begin{array}{rcl}
\PP & \coDefGr  & 
\inact
\mid
\oup\pp{i}{I}{\la}{\PP}%
\mid
\inp\pp{i}{I}{\la}{\PP}%
\end{array}
}
where $I\neq\emptyset$ and $\la_h\neq\la_k$  for  $h, k\in I$ and $h\neq k$.  
\end{definition}

The symbol $ \coDefGr$, in the  above definition  
and in  other definitions, 
indicates that the productions of the grammar should be interpreted \emph{coinductively}.
 That is, they define possibly infinite processes.  
However, we assume such processes to be \emph{regular},  i.e. with 
finitely many distinct subprocesses. In this way, we only obtain processes which are solutions of 
 finite sets  of equations, see \cite{Cour83}.  
We choose this formulation  since 
it allows us to avoid explicitly handling variables, thus simplifying a lot the technical development. 

 A process of shape $\oup\pp{i}{I}{\la}{\PP}$ (\emph{internal choice}) chooses a  label  in the set  $\{\la_i\mid i\in I\}$ to be sent to $\pp$, and then behaves differently depending on the  label sent.  
 A process of shape $\inp\pp{i}{I}{\la}{\PP}$ (\emph{external choice}) waits for receiving one of the  labels  $\{\la_i\mid i\in I\}$ from $\pp$, and then behaves  as $\PP_i$ depending on the
 received label $\la_i$. 
 Note that  the set of indexes in choices is assumed to be non-empty, and the corresponding  labels  to be pairwise distinct.  
An internal choice which is a singleton is simply written $\outpone{\pp}{\la}{\PP}$;
 analogously for an external choice.   The process $\inact$ is inactive and  we omit trailing $\inact$. 
In a full-fledged calculus,  labels would carry values, 
namely they would be of shape $\la(\val)$.  For simplicity, here we
consider ``pure'' labels.  

 The {\em participants of  a process } are the senders and the receivers which occur in  the process itself.   Their set is  defined as the smallest set satisfying 
 
\ClineL{\begin{array}{c}\plays\inact=\emptyset\qquad\plays{\oup\pp{i}{I}{\la}{\PP}}=\plays{\inp\pp{i}{I}{\la}{\PP}}=\set\pp\cup\bigcup_{i\in I}\plays{\PP_i}\end{array}}


We use queues in order  to formalise a one-to-one asynchronous model of communication.
Instead of explicitly  defining  a queue for each possible sender and receiver, we use a 
single queue and equip the communicated labels with their sender and receiver names, so forming
triples that we dub {\em messages}.

\begin{definition}[Messages and Queues]
\begin{enumerate}[i)]
\item
 \emph{Messages} are triples of the form $\mq\pp{\la}\q$ denoting that participant $\pp$ is the sender of  label $\la$  to  the receiver $\q$. 
\item
{\em Message queues} (queues for short) are defined by the following grammar:

\ClineL{
\Msg::=\emptyset \mid
  \addMsg{\mq\pp{\la}\q}{\Msg}
} 
\end{enumerate}
  \end{definition}
  Sent messages are stored in a queue,  from which  they
are subsequently fetched by the receiver.

The order of messages in the queue is the order in which they will be
read. Since 
order matters only 
between messages with
the same sender and receiver, we  always  consider message queues modulo the  following  structural equivalence:\[\addMsg{\addMsg{\Msg}{\mq\pp{\la}\q}}{\addMsg{\mq\pr{\la'}\ps}{\Msg'}}\equiv
  \addMsg{\addMsg{\Msg}{\mq\pr{\la'}\ps}}{\addMsg{\mq\pp{\la}\q}{\Msg'}}
  ~~\text{if}~~\pp\not=\pr~~\text{or}~~\q\not=\ps 
\]
Note, in particular, that
$\addMsg{\mq\pp{\la}\q}{\mq\q{\la'}\pp} \equiv
\addMsg{\mq\q{\la'}\pp}{\mq\pp{\la}\q}$. These two
equivalent queues represent a situation in which both participants
$\pp$ and $\q$ have sent a label 
to the other one, and neither of
them has read the message. This  case  may
happen in a multiparty session with asynchronous communication. 

 The {\em participants of queues} are the senders and the receivers which occur in the queue, i.e.
 
\ClineL{\plays\emptyset=\emptyset\qquad\plays{\mq\pp\la\q\cdot\Msg}=\set{\pp,\q}\cup\plays\Msg}

 A multiparty sessions is  comprised of  a network,
 i.e.  a number of pairs  participant/process  
of shape $\pP{\pp}{\PP}$ composed  in parallel,  each with a different participant $\pp$, and a message queue.

\begin{definition}[Networks and Sessions] 
\begin{enumerate}[i)]
\item
{\em  Networks} are defined as finite parallel composition of named processes, namely

\ClineL{ \Nt = \pP{\pp_1}{\PP_1} \parN
\cdots \parN \pP{\pp_n}{\PP_n}
} 

\noindent where
 $\pp_h \neq \pp_k$  and $\pp_h\not\in\plays{\PP_h}$ for any $1\leq h \neq k\leq n$. 
 \item
{\em  Sessions} are defined as pairs of networks and message queues of the following form: 

\ClineL{\Nt \parallel \Msg}

\end{enumerate}
\end{definition}
 The condition $\pp_h\not\in\plays{\PP_h}$ forbids self-messages. 

We assume the standard structural congruence on  networks  
(denoted $\equiv$), that is we consider sessions modulo permutation of components and adding/removing components of  the  shape $\pP\pp\inact$. 

 If $\PP\neq\inact$ we write $\pP{\pp}{\PP}\in\Nt$ as short for $\Nt\equiv  \pP{\pp}{\PP} \parN\Nt'$ for some $\Nt'$.
This abbreviation is justified by the associativity and commutativity of $\parN$.

 The {\em participants of networks} are the participants which occur in processes, i.e.
 
\ClineL{\plays\Nt=\bigcup_{\pP{\pp}{\PP}\in\Nt}\set{\plays\PP}}

The {\em players of networks} are the  participants 
associated with active processes, i.e.
 
\ClineL{\Plays\Nt=\set{\pp\mid\pP{\pp}{\PP}\in\Nt}}

To define the asynchronous operational semantics of
sessions, we use an LTS whose labels record  the outputs and the inputs. 

\begin{definition}[Asynchronous Operational Semantics]
We equip sessions with the (asynchronous) {\em operational semantics} specified by the LTS of
Figure \ref{fig:asynprocLTS}. Transitions are labelled with 
 \emph{communications}  (ranged over by $\beta$) which are either  the asynchronous emission of a
label  $\M$ from participant $\pp$ to participant $\q$ (notation $\CommAs{\pp}{\M}{\q}$) or 
the actual reading  by participant $\pp$
of the  label  $\M$ sent by participant $\q$ (notation $\CommAsI{\pp}{\M}{\q}$).  
\end{definition}

 \begin{figure}
 \Cline{
\begin{array}{c}  
\\[5pt]
\confAs{\pP{\pp}{\oup\q{i}{I}{\la}{\PP}}\parN\Nt}{\Msg} \stackred{\CommAs\pp{\la_h}{\q}}
  \confAs{\pP{\pp}{\PP_h}\parN\Nt}{\addMsg{\Msg}{\mq\pp{\la_h}{\q}}}\quad \text{ where  }\ 
   h \in I \quad
   {~~~~~~\rn{Send}}
   \\[10pt]
\confAs{\pP{\pp}{\inp\q{i}{I}{\la}{\PP}}\parN\Nt}{\addMsg{\mq{\q}{\la_h}\pp}{\Msg}}\stackred{\CommAsI\pp{\la_h}{\q}}
 \confAs{\pP{\pp}{\PP_h}\parN\Nt}{\Msg}\quad  \text{ where  }\ 
 h \in I  \quad
  {~~~~~~\rn{Rcv}}
   \\[3pt]
\end{array}
} 
\caption{LTS for sessions.}\label{fig:netredAs}\label{fig:asynprocLTS}
\end{figure}

Rule \rn{Send}  in  Figure \ref{fig:asynprocLTS} allows a
participant $\pp$ with an internal choice (a sender) to send one of
its possible  labels  $\la_h$, by adding  the corresponding message  to the
queue. Symmetrically, Rule \rn{Rcv} allows a participant  $\pp$
with an external choice (a receiver) to read the first message 
in the queue sent  to her by a  given  participant $\q$,  if  its label $\la_h$  is one of those she is waiting for.

The {\em players of
   communications} are the senders for the outputs and the receivers for the inputs, i.e. we define
   
\ClineL{
 \play{\CommAs{\pp}{\la}{\q}}=
    \play{\CommAsI{\pp}{\la}{\q}}=\pp}

As usual we define (possibly empty) sequences of communications as traces.
\begin{definition}[Traces]\label{tra}  (Finite) traces are defined by $
\comseqA:=\ee\mid\concat\beta\comseqA$.
\end{definition}
\noindent
When $\comseqA=\concat{\beta_1}{\concat\ldots{\beta_n}}$ ($n\geq 1)$
we write $\Nt\parN\Msg\stackred{\comseqA}\Nt'\parN\Msg'$ as short for

\ClineL{\Nt\parN\Msg\stackred{\beta_1}\Nt_1\parN\Msg_1\cdots\stackred{\beta_n}\Nt_{n}\parN\Msg_{n} 
 =  \Nt'\parN\Msg'}

 With $\Nt\parN \Msg\not\rightarrow$ we denote that the session $\Nt\parN \Msg$ is stuck.

 \begin{example}[A  social media 
 session]\label{smex}
{\em
A social network has two users ($\pu_1$ and $\pu_2$) that want to  interact  
using a service $\ps$. The users exchange messages \msg{go}
and \msg{stop} communicating when they like to continue or not  their interaction.   
They ``should'' \msg{Req}uest \msg{Data} to the service only when they  both are willing to do.  
The above system is roughly described (disregarding the logical order of messages)  in Figure~\ref{fig:compsyst}.
A multiparty session corresponding to this system is the following.

\ClineL{
\pP{\pu_1}{\PU_1} \parN \pP{\pu_2}{\PU_2} \parN \pP{\ps}{\PS} \parN \emptyset
}

\ClineL{
\PU_1 =\pu_2! \left\{\begin{array}{l}
                            \\[-7mm]
                       \msg{go}.\pu_2?\left\{\begin{array}{l}
                              \msg{go}.\ps!\msg{req}.\ps?\msg{data}.\PU_1\\
                              \msg{stop}.\ps!\msg{req}
                                      \end{array}
                            \right.\\
                       \msg{stop}.\pu_2?\left\{\begin{array}{l}
                              \msg{go}\\
                              \msg{stop}
                                      \end{array}
                            \right.\\
                       \end{array}\right.
}

\ClineL{\PU_2 =\pu_1! \left\{\begin{array}{l}
                            \\[-7mm]
                       \msg{go}.\pu_1?\left\{\begin{array}{l}
                              \msg{go}.\ps!\msg{req}.\ps?\msg{data}.\PU_2\\
                              \msg{stop}
                                      \end{array}
                            \right.\\
                       \msg{stop}.\pu_1?\left\{\begin{array}{l}
                              \msg{go}\\
                              \msg{stop}
                                      \end{array}
                            \right.\\
                       \end{array}\right.
}

\ClineL{
 \PS = 
   \pu_2?\msg{req}.\pu_1?\msg{req}. \pu_1!\msg{data}.\pu_2!\msg{data}.\PS 
}\noindent

\noindent
where both participants start sending messages, a feature which  typically can be dealt  only thanks to asynchronous communication.   
The behaviours of $\pu_1$ and $\pu_2$ only differ in  that the process $\PU_1$, after sending \msg{go} to $\pu_2$ and receiving \msg{stop} from $\pu_2$,  sends a $\msg{req}$  to the service.   
So the process $\PU_1$ does not  precisely  implement the prescribed behaviour, while $\PU_2$ does.

\begin{figure*}[t]\centering
 \hspace{0mm}$
\begin{tikzpicture}[node distance=3.2cm,scale=.85,transform shape]
        \node (square-u1) [draw,minimum width=1cm,minimum height=1cm] {\Large $\pu_1$};
        \node (square-s) [draw,minimum width=1cm,minimum height=1cm,  below right of=square-u1, xshift=1.8cm, yshift=8mm] {\Large $\ps$};
        \node (square-u2) at (1,0) [draw,minimum width=1cm,minimum height=1cm,  below of=square-u1, yshift=3mm] {\Large $\pu_2$};

        \path  (square-u1)[-stealth,bend left = 10]  edge node[yshift=0.1cm,right, pos=0.2] {$\msg{req}$} (square-s) 
                   (square-s)[-stealth,bend left = 10]  edge node[yshift=-0.1cm,left, pos=0.2]  {$\msg{data}$} (square-u1)
                    (square-u2)[-stealth,bend left = 10]  edge node[yshift=0.1cm,left, pos=0.3] {$\msg{req}$} (square-s)
                  (square-s)[-stealth,bend left = 10]  edge node[yshift=-0.1cm,right, pos=0.3]  {$\msg{data}$} (square-u2)
                  (square-u1)[-stealth,bend right = 20,pos=0.3]  edge node[right, pos=0.4] {$\msg{go/stop}$} (square-u2)
                   (square-u2)[-stealth,bend left = 70]  edge node[left,pos=0.3]  {$\msg{go/stop}$} (square-u1)       
        ;
    \end{tikzpicture}
$
   \caption{\label{fig:compsyst}Representation of the session of Example~\ref{smex}.}
\end{figure*}
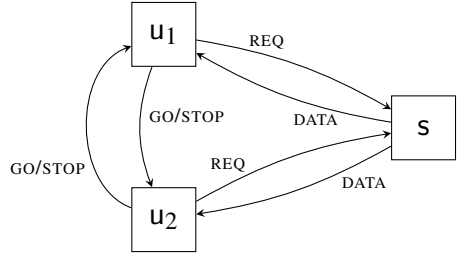

}
\end{example}

\subsection{Partial Communication Properties}

 Now, we define  the property of  $\pas$-lock-freedom.  This property was first introduced in~\cite{BD23}, where $\pas$ was the set of participants whose  lock-freedom we don't care about. 
 $\pas$-lock-freedom  is  a ``partial'' version of  the standard lock-freedom~\cite{Kobayashi02,Padovani14}. 
 The latter consists in the possibility of  completion 
of pending communications of any participant
(this can be alternatively stated by saying that any participant is lock-free).
We are interested instead in the  progress of some explicitly specified participants only.

\begin{definition}[$\pas$-lock-freedom]\label{d:lf}
\begin{enumerate}[i)]
\item 
 A  multiparty session $\Nt\parN\Msg$ is $\pp$-\emph{lock-free} 
if\\ $\Nt\parN\Msg\stackred{\comseqA}\Nt'\parN\Msg'$ and $\pP{\pp}{\PP}\in\Nt'$ imply  
$\Nt'\parN\Msg'\stackred{\concat{\comseqA'}\beta}$
	 for some $\comseqA'$ and $\beta$\\ such that $\pp\in\play\beta$.
	 \item
 A   
 multiparty session  $\Nt\parN\Msg$ is $\pas$-\emph{lock-free} if  it is $\pp$-\emph{lock-free} for each 
$\pp\in\pas$. 
	 \item  A   
	 multiparty session 
	 \emph{$\Nt\parN\Msg$ is a lock-free session} if it is $\pp$-\emph{lock-free} for each 
	  $\pp\in\Plays\Nt$. 
\end{enumerate}
\end{definition}

It is natural to extend also the usual notion of  Deadlock-freedom  to our setting.
\begin{definition}[$\pas$-deadlock-freedom]
  A  
  multiparty session $\Nt\parN \Msg$ is a {\em $\pas$-deadlock-free} session if
	 $\Nt\parN \Msg\stackred{\mypath}\Nt'\parN \Msg'\not\rightarrow$  implies  
	 $\pp\not\in\Plays{\Nt'}$   for any $\pp\in\pas$.
\end{definition}

It is immediate to check that, as for standard  Lock- and Deadlock-freedom,  the following 
hold.
\begin{fact}
$\pas$-lock-freedom  implies $\pas$-deadlock-freedom. 
\end{fact} 
Trivially, as for the standard versions of the properties, the vice versa does not hold whenever $\pas\neq\emptyset$.

\begin{definition}[$\pas$-orphan-message-freedom]\label{def:pomf}\hfill
\begin{enumerate}[i)]
\item\label{def:pomf1}
A  multiparty session $\Nt\parN\Msg$ is $\pp\q$-\emph{orphan-message-free} if
 $\Nt\parN\Msg\stackred{\comseqA}\Nt'\parN\mq\pp\la\q\cdot\Msg'$ implies $\Nt'\parN\mq\pp\la\q\cdot\Msg'\stackred{\comseqA'\cdot\q\pp?\la}$ for some $\comseqA'$.
 \item\label{def:pomf2}
 A  multiparty session $\Nt\parN\Msg$ is $\pas$-\emph{orphan-message-free} if
 it is $\pp\q$-{orphan-message-free} for each pair of participants $\pp,\q\in\pas$.
 \item\label{def:pomf3}
 A  multiparty session $\Nt\parN\Msg$ is \emph{orphan-message-free} if
 it is  $\Plays{\Nt}\cup\plays{\Nt}\cup\plays{\Msg}$-{orphan-message-free}. 
 \end{enumerate}
\end{definition}
 Point~(\ref{def:pomf3} of  previous  definition is justified by the example $\Nt=\pP\pp{\q!\lambda}\parN\emptyset\stackred{\pp\q!\la} \pP\pp\inact\parN\mq\pp\la\q$, where the message $\mq\pp\la\q$ is orphan and $\pp\in\Plays\Nt$, $\q\in\plays\Nt$. 

\begin{example}[Partial properties for the social media example]
\em
It is not difficult to check that the session of Example~\ref{smex} is neither lock-free nor
orphan-message-free.
In fact we get an  $\ps$-lock 
whenever  at least one among  $\pu_1$ and  $\pu_2$ sends to  the other the message $\msg{stop}$.  In such a case the process
of  $\ps$  is not $\inact$, but unable to perform the input action it is willing to do.
An orphan message does result present in the queue because of a ``programming error'':
in case   $\pu_1$  sends $\msg{go}$ to $\pu_2$, receives $\msg{stop}$ from $\pu_2$ and then sends $\msg{req}$ to the server, 
 it happens that  such a $\msg{req}$ from  $\pu_1$ will never be received by $\ps$, since 
 a $\msg{req}$ from $\pu_2$ should be received first, but such a message will never be sent.  

The social network, however, is interested in the absence of locks for the $\set{\pu_1,\pu_2}$ subsystem only  (i.e. $\set{\pu_1,\pu_2}$-lock-freedom) 
as well in the absence of orphan-messages only for  the messages  exchanged among $\pu_1$ and $\pu_2$  (i.e. $\set{\pu_1,\pu_2}$-orphan-message-freedom). 
\end{example}

 \section{Global Types and Type System 
 }\label{ts}

The vast majority of global types used in the literature are independent of the
synchronicity/asynchronicity of the underlying communication model.
 This means that, in a global type, the exchange of a message $\mathsf{m}$ from a
participant $\mathtt{A}$ to a participant $\mathtt{B}$ is generally represented by something 
like $\mathtt{A}\stackrel{\mathsf{m}}{\longrightarrow}\mathtt{B}$. This is then interpreted
either as the synchronous exchange of $\mathsf{m}$ 
according to a handshaking protocol between   $\mathtt{A}$ and $\mathtt{B}$ or as 
the simultaneous representation of two distinct asynchronous actions: the insertion of  
$\mathsf{m}$ in a communication medium (typically a queue or a bag) and the acquisition
of the message from that. 
In \cite{CDG21,CDG22,DGD22} global types are instead strictly tailored for 
asynchronous interactions: 
the separate output and input actions, which together form an  asynchronous  communication 
 (respectively $\pp\q!\lambda$ and $\pp\q?\lambda$ in our formalism, see below), 
are made visible in the global type.
Even if this is actually more than  what  a choreographic formalism should require
 (our one can in fact hardly be considered a choreographic formalism in the usual sense),
it allows the global types to be used in a type assignment system for
asynchronous processes guaranteeing relevant (partial, in our case) communication properties.
 Being the asynchrony of communication syntactically evident in the global type,
the formal verification of such properties can be performed without having to consider a layer of ``semantic'' interpretation of the types, so maintaining the complexity of proofs at the same
complexity level as those for synchronous formalisms like the one in \cite{BD23}.

\begin{definition}[Asynchronous Global  Types] 
\label{gt}
 \emph{(Asynchronous) global types} 
$\G$ are defined by the following grammar:
 
\ClineL{\begin{array}{rcl}
\G & \coDefGr &    \agtO{\pp}{\q}i I{\la}{\G}
              \mid \agtI \pp\q i I \la \G   
              \mid \End
\end{array}}

\noindent
where $I\neq\emptyset$, $\pp\neq \q$ and
$\la_h\not=\la_k\,$ 
for $h, k\in I$ and $h\neq k$. 
\end{definition}

 As for processes, $ \coDefGr$ indicates that global types are
coinductively defined \emph{regular} terms.  The global type
$\agtO{\pp}{\q}i I{\la}{\G}$ specifies that $\pp$ sends a label
$\la_h$ with $h\in I$ to $\q$ and then the interaction described by
the global type $\G_h$ takes place. Dually, the global type $\agtI
\pp\q i I \la \G $ specifies that $\q$ receives label $\la$ from $\pp$ and
then the interaction described by the global type $\G$ takes place. The terminated global type is $\End$ and we
will omit trailing $\End$'s. 

Clearly  message outputs   must precede the corresponding inputs, since in the asynchronous communication the output puts the message on the queue and the input takes the message from the queue. Once  a message  
is on the queue no other  message  
can be read  with  
the same sender and receiver. This justifies the fact that inputs in global types have no choices. 

\begin{example}[A global type for the social media example] {\em 
A global type describing a possible behaviour of the network of Example \ref{smex}
is provided in Figure \ref{fig:gtsm}.}
\begin{figure}[h]
\ClineL{
\G =\pu_1\pu_2! \left\{\begin{array}{l}
                       \\[-9mm]
                       \msg{go}.\pu_2\pu_1! \left\{\begin{array}{l}
                                   \msg{go}.
                                           \pu_1\pu_2?\msg{go}.\pu_2\pu_1?\msg{go}.\pu_2\ps!\msg{req}
                                           .\ps\,\pu_2?\msg{req}.
                                           \pu_1\ps!\msg{req}.\ps\,\pu_1?\msg{req}.\\[-3mm]
                                                 \hspace{83mm}      \hookleftarrow\\[-3mm]
                                          \hspace{5.3mm} \ps\,\pu_1!\msg{data}.\pu_1\,\ps?\msg{data}.
                                          \ps\,\pu_2!\msg{data}. \pu_2\,\ps?\msg{data}. 
                                          \G
                                  \\
                                  \msg{stop}.\pu_1\pu_2?\msg{stop}.\pu_2\pu_1?\msg{go}.\pu_1\ps!\msg{req}
                                                                     \end{array}\right.  \\[7mm]
                             \vspace{-2mm}\msg{stop}.\pu_2\pu_1!\left\{\begin{array}{l}
                                                             \msg{go}.\pu_1\pu_2?\msg{go}.\pu_2\pu_1?\msg{stop}\\
                                                             \msg{stop}.\pu_1\pu_2?\msg{stop}.\pu_2\pu_1?\msg{stop}
                                                                          \end{array}\right.
                                   \end{array}\right.
}
\caption{A global type for the social media session.}\label{fig:gtsm} 
\end{figure}
\end{example}

The set of {\em players of a global type,} notation 
$\Plays{\G}$, is  
the smallest set satisfying the following equations: 

\ClineL{\begin{array}{c}
\Plays{\End} = \emptyset \\
\Plays{ \agtO{\pp}{\q}i I{\la}{\G}} = \set{\pp}\cup\bigcup_{i\in I}\Plays{\G_i} \qquad
 \Plays{\agtI \pp\q  i I\la {\G'}} = \set{\pp}\cup\Plays{\G'} 
\end{array}}

\noindent
Notice that the sets of players are always finite thanks to the regularity of global types. 

To guarantee good communication properties for typable sessions, we require global types to satisfy a boundedness condition. 
To formalise boundedness we use the notion of  \emph{path} of a global type.
{\em Paths} are actual paths in global types  viewed  as trees.
They are possibly infinite sequences of communications, and are ranged over by $\ipth$.
Note that a finite path is a trace in the sense of \refToDef{tra}. 
 We extend the notation $\cdot$ to denote 
also  the concatenation of a finite sequence with a possibly
infinite sequence.  
The function ${\sf Paths}$ returns  the  set of all the
\emph{paths} of a global type and is defined as  the greatest set such that:

\ClineL{\begin{array}{ll} 
\IPaths{\End} &= \set{\epsilon}  \\ 
\IPaths{\agtO{\pp}{\q}i I{\la}{\G}} &= \bigcup_{i\in I} \set{ \concat{\CommAs\pp{\la_i}\q}{\ipth} \mid \ipth\in\IPaths{\G_i} } \\
 \IPaths{\agtI \pp\q i I \la {\G'} } 
&= 
\set{ \concat{\CommAsI\pp{\la}\q}{\ipth} \mid \ipth\in \IPaths{\G'} } 
\end{array}} 

 If $x \in \mathbf{N} \cup \set{\infty}$ is the
length of $\ipth$, i.e. $x=\mid{\ipth}\mid$, we denote by $\ipth[n]$ the $n$-th communication in the path $\ipth$,
where $1\le n < x$ if $x=\infty$ and $1\le n \le x$ if $x\neq\infty$. 
It is handy to define the \emph{depth} of a player $\pp$ in a global type 
$\G$, $\weight(\G,\pp)$. 

\begin{definition}[Depth of a  Player]\label{def:depth}
Let $\G$ be a global type. 
For ${\ipth\in\IPaths{\G}}$ set

 \ClineL{\weight(\ipth,\pp) = \inf \{ n \mid \play{\ipth[n]} = \pp \}}
 
and define $\weight(\G,\pp)$, the \emph{depth} of $\pp$ in $\G$, as follows:

\ClineL{\weight(\G,\pp) = \begin{cases} 
\sup \{ \weight(\ipth,\pp) \mid \ipth \in \IPaths{\G} \}& \pp \in \Plays{\G} \\ 
0 & \text{otherwise} 
\end{cases}
}

\end{definition}
Note that $\weight(\G,\pp)=0$ iff  $\pp \not\in \Plays{\G}$. 
 Moreover, if  $\pp\ne\play{\ipth[n]}$  for 
all $n\in\mathbf{N}$,  then $\weight(\ipth,\pp) = \inf\, \emptyset = \infty$. 
Hence, if $\pp$ is a player of a global type $\G$ 
and there is some path in $\G$ where $\pp$ does not occur as a player, 
then $\weight(\G,\pp) = \infty$.

\begin{definition}[Boundedness]\label{def:bound}
A global type $\G$ is \emph{bounded} if  $\weight(\G',\pp)$ is finite
for each participant $\pp\in\Plays{\G}$ and  each type  
$\G'$ which occurs in   $\G$. 
\end{definition}
\begin{example}\label{exb} {\em The following example shows the  necessity  of considering all types  occurring in a global type for defining boundedness. 
Consider 
 $\G= \agtoneO \pr\q {\la}{\agtIS\q\pr  {\la}.\G'} $,  where
 
\ClineL{\G'=\agtSOS \pp\q {\{\Seq{\la_1}{\agtIS \q \pp {\la_1}.\agtoneO \q\pr {\la_3}\agtIS\pr\q {\la_3}}\, ,\,\Seq{\la_2}{\Seq{\agtIS \q\pp {\la_2}}{\G'}}\}}
} 

\noindent Then we have: $
 \weight(\G,\pp)=3,\weight(\G,\q)=2 ,\weight(\G,\pr)=1
$, 
whereas 
$\weight(\G',\pp)=1,\weight(\G',\q)=2,\weight(\G',\pr)=\infty
$.}
\end{example}
Since global types are regular, the boundedness condition is decidable. 

The following notion of {\em weight} will be  used  
for defining the subsequent notion of
$\pas$-soundness, a condition in the typing rules, needed to guarantee $\pas$-orphan-message-freedom.  The weight says  if and  where the global type prescribes an input corresponding to a message. Clearly if the message is $\mq\pp\la\q$ and the global type is $\q\pp?\la'.\G$ with $\la\neq\la'$, then the global type forbids to read this message. 
 
\begin{definition}[Weight]\label{wm}$\;$

\ClineL{\wgs{\mq\pp\la\q}\G = \begin{cases}
0      & \text{if $\G = \q\pp?\la.\G'$} \\ 
\infty & \text{if $\G = \End$ or $\G=\q\pp?\la'.\G'$ with $\la\neq\la'$} \\ 
1 + \max_{i\in I} \wgs{\mq\pp\la\q}{\G_i} & \text{if $\G = \agtO{\pr}{\ps}{i}{I}{\la}{\G}$}\\
1 + \wgs{\mq\pp\la\q}{\G'}&\text{if $\G = \pr\ps?\la'.\G'$ and $\pr\neq\pp$ or $\ps\neq\q$}
\end{cases}}
\end{definition}

 We consider the parallel composition of a global type with a queue that we dub {\em type configuration}. The $\pas$-soundness of type configurations ensures that all messages with both participants in $\pas$ have corresponding inputs in all the paths of the global type. 

\begin{definition}[$\pas$-soundness]\label{ps}
A type configuration $\G\parG\Msg$ is $\pas$-{\em sound} if  $\wgs{\mq\pp\la\q}\G$ is finite for all messages $\mq\pp\la\q$ which occur in $\Msg$ with $\set{\pp,\q}\subseteq\pas$. \end{definition}

\subsection{Partial Type System}
 
 As mentioned before, we devise a type system ensuring partial communication properties
 for typable sessions.
 Being in an asynchronous setting, some restrictions have to be imposed in order to guarantee
 decidability of typability. We achieve that by looking at  queues as invariants for cycles.
 This is a quite more flexible condition than, for instance, imposing a fixed bound on the number of messages between participants. It would be rather cumbersome to guarantee our condition
 in a coinductive type system which, like those in~\cite{CDG21,CDG22,DGD22,BD23}, suits a formalism with coinductively defined processes and types.  
 We hence introduce an implicitly coinductive type system, that is looking like the 
 inductive versions of coinductive systems, as defined in~\cite[Section  21.9]{pier02}.
 We define an inductive system with {\em histories} (see below), where the queue invariance
 can be immediately guarantee by the typing rule for cycles.

 \begin{definition}[Histories]
  A {\em history} $\Hset$ is a finite  set  of (session, global type) pairs, namely
 
\ClineL{\Hset::=\emptyset\mid\Hset,\pair{\Nt\parN\Msg}\G}

\end{definition}

We define 
$\pair{\Nt\parN-}\G\in\Hset$ if $\pair{\Nt\parN\Msg}\G\in\Hset$  
for some $\Msg$.

\begin{definition}[Partial Type System]
The judgements of our partial type system have the form

\ClineL{
\Hset\tyng{\pas}{\G}{\Nt\parN\Msg}
}

\noindent
where $\pas$ is a set of participants (those whose properties we are interested in)
and where the global type $\G$ is bounded. 
The inference rules are described in Figure \ref{fig:cntr}. 
\end{definition}

\begin{figure}[h]
\begin{math}
\begin{array}{c}
\NamedRule{\rn{End}}{\text{ $\Plays{\Nt}\cap\pas=\emptyset$}}{ \Hset\tyng\pas\End{\Nt\parN\Msg} }{\End\parG\Msg~\mathrm{is}~\pas\text{-}\mathrm{sound}} \qquad\qquad
\NamedRule{\rn{Cycle}}{\text{ $\pair{\Nt\parN\Msg}\G\in\Hset$}}{ \Hset
\tyng\pas\G{\Nt\parN\Msg} }{} \\[6ex] 
\NamedRule{\rn{Out}}{
\Hset,\pair{\pP\pp{\PP}\parN\Nt\parN\Msg}\G  \tyng\pas{\G_i}{\pP\pp{\PP_i}\parN\Nt\parN\Msg\cdot\mq\pp{\la_i}{\q}}\\ \text{ $(\Plays{\Nt}\setminus\Plays{\G})\cap\pas=\emptyset\quad \forall  i  \in I $}
}{ \Hset\tyng\pas{\G}{\pP\pp{\PP}\parN\Nt\parN\Msg}}
{\begin{array}{c}\G\parG\Msg~\mathrm{is}~\pas\text{-}\mathrm{sound}\\
\pair{\pP\pp{\PP}\parN\Nt\parN-}\G\not\in\Hset\\
\G=\agtO{\pp}{\q}{i}{I}{\la}{\G}\quad
\PP=\oup{\q}{i}{I}{\la}{\PP}
\end{array}}
\\[6ex] 
\NamedRule{\rn{In}}{
\Hset,\pair{\pP\pp{\PP}\parN\Nt\parN\Msg}\G  \tyng\pas{\G'}{\pP\pp{\PP_{h}}\parN\Nt\parN\Msg}\\
  \text{ $(\Plays{\Nt}\setminus\Plays{\G})\cap\pas=\emptyset\quad h  \in I$}
}{ \Hset\tyng\pas{\G}{\pP\pp{\PP}\parN\Nt\parN\mq\q{\la_h}\pp\cdot\Msg} }
{\begin{array}{c}\G'\parG\Msg~\mathrm{is}~\pas\text{-}\mathrm{sound}\\
\pair{\pP\pp{\PP}\parN\Nt\parN-}\G\not\in\Hset\\\G=\agtI{\pp}{\q}i I{\la_h}{\G'}\quad\PP=\inp{\q}{i}{I}{\la}{\PP}\end{array}}
\end{array} 
\end{math}
\caption{Typing rules for sessions.}\label{fig:cntr} 
\end{figure}
\begin{figure}[h]\centerline{\small 
\prooftree
      \prooftree
            \prooftree
                 \prooftree
                      \prooftree
                           \prooftree
                               \prooftree
                                    \prooftree
                                        \prooftree
                                            \prooftree
                                                \prooftree
                                                    \prooftree
                                                       \prooftree
                                                       \justifies
                                                            \Hset_{15} \vdash_\pas
                                                    \pP{\pu_1}{\PU_1} \parN \pP{\pu_2}{\PU_2} \parN \pP{\ps}{\PS} \parN \emptyset   :\G 
                                                       \using{\rn{\small cycle}}
                                                       \endprooftree
                                                    \justifies
                                                        \Hset_{14} \vdash_\pas
                                                 \pP{\pu_1}{\PU_1} \parN \pP{\pu_2}{\PU^{\text{\tiny IV}}_2 } \parN \pP{\ps}{\PS} \parN     \Msg_6  :\G_{15}
                                                    \endprooftree
                                                \justifies
                                                     \Hset_{13} \vdash_\pas
                                                 \pP{\pu_1}{ \PU_1  } \parN \pP{\pu_2}{\PU^{\text{\tiny IV}}_2 } \parN \pP{\ps}{ \PS^{\text{\tiny III}}} \parN    \emptyset  
                                                 :\G_{14}
                                                \endprooftree
                                            \justifies
                                                \Hset_{12} \vdash_\pas
                                                 \pP{\pu_1}{\PU^{\text{\tiny IV}}_1} \parN \pP{\pu_2}{\PU^{\text{\tiny IV}}_2 } \parN \pP{\ps}{\PS^{\text{\tiny III}}} \parN    \Msg_5 :\G_{13}
                                            \endprooftree
                                        \justifies
                                             \Hset_{11} \vdash_\pas
                                                 \pP{\pu_1}{\PU^{\text{\tiny IV}}_1} \parN \pP{\pu_2}{\PU^{\text{\tiny IV}}_2 } \parN \pP{\ps}{\PS^{\text{\tiny II}}} \parN   \emptyset :\G_{12}
                                        \endprooftree
                                    \justifies
                                       \Hset_{10} \vdash_\pas
                                                \pP{\pu_1}{\PU^{\text{\tiny IV}}_1} \parN \pP{\pu_2}{\PU^{\text{\tiny IV}}_2 } \parN \pP{\ps}{\PS^{\text{\tiny I}}} \parN    \Msg_4 :\G_{11}
                                    \endprooftree
                               \justifies
                                  \Hset_9 \vdash_\pas
                                                  \pP{\pu_1}{ \PU^{\text{\tiny III}}_1} \parN \pP{\pu_2}{\PU^{\text{\tiny IV}}_2 } \parN \pP{\ps}{\PS^{\text{\tiny I}}} \parN  \emptyset :\G_{10}
                               \endprooftree
                           \justifies
                               \Hset_7 \vdash_\pas
                               \pP{\pu_1}{ \PU^{\text{\tiny III}}_1} \parN \pP{\pu_2}{ \PU^{\text{\tiny IV}}_2 } \parN \pP{\ps}{\PS} \parN  \Msg_3 :\G_9
                           \endprooftree
                      \justifies
                      \Hset_5 \vdash_\pas
                                \pP{\pu_1}{\PU^{\text{\tiny III}}_1} \parN \pP{\pu_2}{\PU^{\text{\tiny III}}_2} \parN \pP{\ps}{\PS} \parN \emptyset  :\G_7
                      \endprooftree
                 \justifies
                    \Hset_3 \vdash_\pas
                                  \pP{\pu_1}{\PU^{\text{\tiny III}}_1} \parN \pP{\pu_2}{\PU^{\text{\tiny I}}_2} \parN \pP{\ps}{\PS} \parN \Msg_0     : \G_5
                 \endprooftree
            \justifies
              \Hset_2 \vdash_\pas
                             \pP{\pu_1}{\PU^{\text{\tiny I}}_1} \parN \pP{\pu_2}{\PU^{\text{\tiny I}}_2} \parN \pP{\ps}{\PS} \parN \Msg_1 : \G_3
            \endprooftree
            \hspace{-4mm}
            \prooftree
                 \prooftree
                     \prooftree
                          \prooftree
                          \justifies
                                \Hset_{8} \vdash_\pas
                               \pP{\pu_1}{\inact} \parN \pP{\pu_2}{\inact} \parN \pP{\ps}{\PS} \parN \Msg_7 :\End
                         \using \rn{\small end}
                          \endprooftree
                     \justifies
                            \Hset_6 \vdash_\pas
                                 \pP{\pu_1}{\PU^{\text{\tiny V}}_1} \parN \pP{\pu_2}{\inact} \parN \pP{\ps}{\PS} \parN \emptyset :\G_8                    
                     \endprooftree
                 \justifies
                    \Hset_4 \vdash_\pas
                            \pP{\pu_1}{\PU^{\text{\tiny V}}_1} \parN \pP{\pu_2}{\PU^{\text{\tiny II}}_2} \parN \pP{\ps}{\PS} \parN \Msg_0    :\G_6
                 \endprooftree
            \justifies
              \Hset_2 \vdash_\pas
                             \pP{\pu_1}{\PU^{\text{\tiny I}}_1} \parN \pP{\pu_2}{\PU^{\text{\tiny II}}_2} \parN \pP{\ps}{\PS} \parN \Msg_2 :\G_4
            \endprooftree
      \justifies
      \Hset_1 \vdash_\pas 
                    \pP{\pu_1}{\PU^{\text{\tiny I}}_1} \parN \pP{\pu_2}{\PU_2} \parN \pP{\ps}{\PS} \parN \Msg_0 : \G_1
      \endprooftree
      \text{\Large $\mathcal{D}$}\quad
\justifies
\vdash_\pas 
        \pP{\pu_1}{\PU_1} \parN \pP{\pu_2}{\PU_2} \parN \pP{\ps}{\PS} \parN \emptyset : \G
\endprooftree
}
\caption{Derivation for the social media example.}\label{der}
\end{figure}

In case all the participants in $\pas$ (those we care about)  terminate, we are not interested
 anymore in  what  other participants do and hence we do not record their behaviours in the
 global type. This is essentially what is formalised by Axiom $\rn{End}$.  
 No message with both sender and receiver in $\pas$ must be present in the queue
 if we wish to  ensure $\pas$-orphan-message-freedom.
 This is formalised by the clause ``$\End\parG\Msg~\mathit{is}~\pas\text{-}\mathit{sound}$''
of   Axiom   \rn{End}. 
 
 The inductive rules of our system can be looked at as a type reconstruction algorithm for
 a coinductively defined system. 
 
 We formalise in  
   Axiom   $\rn{Cycle}$ also an invariant requirement for  ensuring 
  decidability,
 namely  the  invariance of queues for cycles. This implies that any output in a cycle must have
 a corresponding input in the cycle itself. 
 
 Rules $\rn{Out}$ and $\rn{In}$ enable to record in the global types the actions performed
 by processes.\\ 
  Rule $\rn{Out}$ adds in the process and in the global type the same outputs. Rule $\rn{In}$ adds one input in the global type and it allows more inputs in the process, mimicking the subtyping for session types~\cite{DH12}.\\ Both rules require as premises the typability of the sessions obtained by reducing the added communications. These rules ask for some conditions. 
 The condition $(\Plays{\Nt}\setminus\Plays{\G})\cap\pas=\emptyset$ ensures that the communications done by players in $\Nt$ which belong to $\pas$ are recorded in $\G$. 
  The $\pas$-soundness condition for configurations is needed to  ensure 
 absence of orphan-messages with sender and receiver in $\pas$.
 The condition $\pair{\pP\pp{\PP}\parN\Nt\parN-}\G\not\in\Hset$, together with the one for
 Axiom $\rn{Cycle}$, is used for ensuring decidability. 
 Our type system is in fact decidable, since global types and processes are regular. 
In particular, any bottom-up attempt  to reconstruct  
a branch of a to-be derivation necessarily 
ends up with an application of Axiom $\rn{End}$, or of Axiom $\rn{Cycle}$ or fails
because Rules $\rn{Out}$ and $\rn{In}$ do not apply.

Whereas our type system enables to deal with participants whose lock-freedom
we do care about, the system of \cite{BD23}, besides taking into account a synchronous
model of communication, deals with participants whose lock-freedom
we do not care about. 
Even if equivalent from an abstract viewpoint, these two different perspectives from which
one can deal with the notion of  ``partiality'' ,
bring with them pros and cons when formalised in specific MPST type systems.
For instance, something like the rule  $\rn{Weak}$ of  \cite{BD23}
is not needed here, so accounting for simpler proofs. 
On the other hand, the loose treatment of disregarded participants in \cite{BD23},
where one can consider different sets of participants 
in different branches of derivations, allows  for a modular development of the derivations. 

The presence of queues in our asynchronous setting 
makes some extra conditions -- besides the regularity of global types and processes --
necessary in order to get a decidable type systems.
Such extra conditions are definitely easier to formalise in an inductive system rather that in
a coinductive one, so accounting for the use of an inductive system, unlike a coinductive one as in \cite{BD23}.

\begin{example}[Typing for the social media example]\em
\begin{figure}
 \centerline{$
\begin{array}{l}
\Hset_1 = (\pP{\pu_1}{\PU_1} \parN \pP{\pu_2}{\PU_2} \parN \pP{\ps}{\PS} \parN \emptyset, \G)\\[5pt]
\Hset_2 = \Hset_1,(\pP{\pu_1}{\PU^{\text{\tiny I}}_1} \parN \pP{\pu_2}{\PU_2} \parN \pP{\ps}{\PS} \parN \Msg_0, \G_1)\\[5pt]
\Hset_3 = \Hset_2, (\pP{\pu_1}{\PU^{\text{\tiny I}}_1} \parN \pP{\pu_2}{\PU^{\text{\tiny I}}_2} \parN \pP{\ps}{\PS} \parN \Msg_1, \G_3)\\[5pt]
\Hset_4 =  \Hset_2, (\pP{\pu_1}{\PU^{\text{\tiny I}}_1} \parN \pP{\pu_2}{\PU^{\text{\tiny II}}_2} \parN \pP{\ps}{\PS} \parN \Msg_2,\G_4)\\[5pt]
\Hset_5 =  \Hset_3, (\pP{\pu_1}{\PU^{\text{\tiny III}}_1} \parN \pP{\pu_2}{\PU^{\text{\tiny I}}_2} \parN \pP{\ps}{\PS} \parN \Msg_0, \G_5)\\[5pt]
\Hset_6 =  \Hset_4, (\pP{\pu_1}{\PU^{\text{\tiny V}}_1} \parN \pP{\pu_2}{\PU^{\text{\tiny II}}_2} \parN \pP{\ps}{\PS} \parN \Msg_0,\G_6)\\[5pt]
\Hset_7 = \Hset_5, (\pP{\pu_1}{\PU^{\text{\tiny III}}_1} \parN \pP{\pu_2}{\PU^{\text{\tiny III}}_2} \parN \pP{\ps}{\PS} \parN \emptyset, \G_7)\\[5pt]
\Hset_8 =  \Hset_6, (\pP{\pu_1}{\PU^{\text{\tiny V}}_1} \parN \pP{\pu_2}{\inact} \parN \pP{\ps}{\PS} \parN \emptyset,\G_8)\\[5pt]
\Hset_9 = \Hset_7, (\pP{\pu_1}{ \PU^{\text{\tiny III}}_1 } \parN \pP{\pu_2}{ \PU^{\text{\tiny IV}}_2  } \parN \pP{\ps}{\PS} \parN  \Msg_3,\G_9)\\[5pt]
\Hset_{10} = \Hset_9, (\pP{\pu_1}{ \PU^{\text{\tiny III}}_1 } \parN \pP{\pu_2}{ \PU^{\text{\tiny IV}}_2 } \parN \pP{\ps}{\PS^{\text{\tiny I}}} \parN  \emptyset, \G_{10})\\[5pt] 
\Hset_{11} = \Hset_{10}, (\pP{\pu_1}{\PU^{\text{\tiny IV}}_1} \parN \pP{\pu_2}{\PU^{\text{\tiny IV}}_2 } \parN \pP{\ps}{\PS^{\text{\tiny I}}} \parN    \Msg_4,\G_{11})\\[5pt]
\Hset_{12} = \Hset_{11}, ( \pP{\pu_1}{\PU^{\text{\tiny IV}}_1} \parN \pP{\pu_2}{\PU^{\text{\tiny IV}}_2 } \parN \pP{\ps}{\PS^{\text{\tiny II}}} \parN   \emptyset, \G_{12})\\[5pt]
\Hset_{13} = \Hset_{12}, (\pP{\pu_1}{\PU^{\text{\tiny IV}}_1} \parN \pP{\pu_2}{\PU^{\text{\tiny IV}}_2 } \parN \pP{\ps}{\PS^{\text{\tiny III}}} \parN    \Msg_5,\G_{13})\\[5pt]
\Hset_{14} = \Hset_{13}, (\pP{\pu_1}{ \PU_1 } \parN \pP{\pu_2}{\PU^{\text{\tiny IV}}_2 } \parN \pP{\ps}{ \PS^{\text{\tiny III}}} \parN    \emptyset ,\G_{14})\\[5pt]
\Hset_{15} = \Hset_{14}, (\pP{\pu_1}{\PU_1} \parN \pP{\pu_2}{\PU^{\text{\tiny IV}}_2 } \parN \pP{\ps}{\PS} \parN     \Msg_6  :\G_{15})
\end{array}
$} 
\caption{Histories for the derivation of the social media example.}\label{h}
\end{figure}

\begin{figure}
\centerline{$
\begin{array}{l}
\PU^{\text{\tiny I}}_1 = \pu_2?\left\{\begin{array}{l}
                              \msg{go}.\PU^{\text{\tiny III}}_1\\
                              \msg{stop}.\PU^{\text{\tiny V}}_1
                                      \end{array}
                            \right.\quad
 \PU^{\text{\tiny II}}_1 = \pu_2?\left\{\begin{array}{l}
                              \msg{go}\\
                              \msg{stop}
                                      \end{array}
                            \right.\quad
\PU^{\text{\tiny III}}_1 = \ps!\msg{req}.\PU^{\text{\tiny IV}}_1\quad
\PU^{\text{\tiny IV}}_1 = \ps?\msg{data}.\PU_1\quad
\PU^{\text{\tiny V}}_1 =  \ps!\msg{req}\\[20pt]                           
\PU^{\text{\tiny I}}_2 = \pu_1?\left\{\begin{array}{l}
                              \msg{go}.\PU^{\text{\tiny III}}_2\\
                              \msg{stop}
                                      \end{array}
                            \right.\qquad
 \PU^{\text{\tiny II}}_2 = \pu_1?\left\{\begin{array}{l}
                              \msg{go}\\
                              \msg{stop}
                                      \end{array}
                            \right.\qquad
\PU^{\text{\tiny III}}_2 = \ps!\msg{req}.\PU^{\text{\tiny IV}}_2\qquad
\PU^{\text{\tiny IV}}_2 = \ps?\msg{data}.\PU_2\\[20pt]                                                
\PS^{\text{\tiny I}} = \pu_1?\msg{req}.\PS^{\text{\tiny II}}\qquad
\PS^{\text{\tiny II}} = \pu_1!\msg{data}.\PS^{\text{\tiny III}}\qquad
\PS^{\text{\tiny III}} =  \pu_2!\msg{data}.\PS 
\end{array}
$}

\caption{Processes for the derivation of the social media example.}\label{g}
\end{figure}

\begin{figure}
\centerline{$
\begin{array}{l}
\Msg_0=\mq{\pu_1}{\msg{go}}{\pu_2}\qquad
\Msg_1=\Msg_0\cdot \mq{\pu_2}{\msg{go}}{\pu_1}\qquad
\Msg_2=\Msg_0\cdot \mq{\pu_2}{\msg{stop}}{\pu_1}\qquad
\Msg_3=\mq{\pu_2}{\msg{req}}{\ps}\\[5pt]
\Msg_4=\mq{\pu_1}{\msg{req}}{\ps}\qquad
\Msg_5=\mq{\ps}{\msg{data}}{ \pu_1  }\qquad
\Msg_6= 
 \mq{\ps}{\msg{data}}{\pu_2}\qquad
\Msg_7=\mq{\pu_1}{\msg{req}}{\ps}
\end{array}
$}


\caption{Queues for the derivation of the social media example.}\label{q}
\end{figure}

\begin{figure}[t]

\centerline{$
\begin{array}{l}
\G_1 = \pu_2\pu_1! \left\{\begin{array}{l}
                                   \msg{go}.\G_{3}
                                           \\
                                  \msg{stop}.\G_4
                                                                     \end{array}\right.  \qquad
\G_2 = \pu_2\pu_1!\left\{\begin{array}{l}
                                                             \msg{go}.\pu_1\pu_2?\msg{go}.\pu_2\pu_1?\msg{stop}\\
                                                             \msg{stop}.\pu_1\pu_2?\msg{stop}.\pu_2\pu_1?\msg{stop}
                                                                          \end{array}\right.\qquad                                                                     
\G_{3} = \pu_1\pu_2?\msg{go}.\G_{5}\\[20pt]

\G_{4} = \pu_1\pu_2?\msg{stop}.\G_{6}\qquad

\G_{5} = \pu_2\pu_1?\msg{go}.\G_{7}\qquad

\G_{6} = \pu_2\pu_1?\msg{go}.\G_{8}\qquad

 \G_{7} = \pu_2\ps!\msg{req}.\G_{9}\\[20pt]

\G_{8} = \pu_1\ps!\msg{req}\qquad

 \G_{9} =  \ps\,\pu_2?\msg{req}.\G_{10}\qquad
\G_{10} =  \pu_1\ps!\msg{req}.\G_{11}\qquad

\G_{11} = \ps\,\pu_1?\msg{req}.\G_{12}\\[20pt]

 \G_{12}  =  \ps\,\pu_1!\msg{data}. \G_{13} \qquad
 \G_{13}  =   \pu_1\,\ps?\msg{data}. \G_{14}\qquad

 \G_{14}  = \ps\,\pu_2!\msg{data}. \G_{15}\qquad

\G_{15} =  \pu_2\,\ps?\msg{data}. \G\end{array}
$}

\caption{Global types for the derivation of the social media example.}\label{fp}
\end{figure}

\noindent
The type derivation for our social media example is shown in Figure~\ref{der} 
where $\pas={\Set{\pu_1,\pu_2}}$ and  $\mathcal{D}$ is the derivation with conclusion

\ClineL{\Hset_1 \vdash_{\Set{\pu_1,\pu_2}} \pP{\pu_1}{\PU^{\text{\tiny II}}_1} \parN  \pP{\pu_2}{\PU_2} \parN \pP{\ps}{\PS} \parN \mq{\pu_1}{\msg{stop}}{\pu_2} : \G_2}
\noindent
whose detailed description we omit for the sake of readability.  
The abbreviations used in Figure~\ref{der} are listed in Figures \ref{h}, \ref{g}, \ref{q}, 
\ref{fp}. 
\end{example}

\begin{figure}
\begin{math} 
\begin{array}{c}
\NamedRule{{\rn{Top-Out}}}
 { } 
 {\agtO{\pp}{\q}i I{\la}{\G}\parG\Msg \stackred{\CommAs\pp{\la_h}{\q}}\G_h\parG\addMsg\Msg{\mq\pp{\la_h}{\q}}}
 {h\in I}
\\[4ex]
\NamedRule{{\rn{Top-In}}}
 { } 
 {\agtI \pp\q i I\la \G\parG\addMsg{\mq{\q}{\la}\pp}\Msg \stackred{\CommAsI{\pp}{\la}\q}\G\parG\Msg}
 {}
\\[4ex]
\NamedRule{{\rn{Inside-Out}}}
 {\G_i\parG \Msg\cdot\mq\pp{\la_i}{\q} \stackred\asCom\G'_i \parG \Msg'\cdot\mq\pp{\la_i}{\q} \quad \forall i \in I} 
 {\agtO{\pp}{\q}i I{\la}{\G}\parG\Msg \stackred \asCom\agtO{\pp}{\q}i I{\la}{\G'}\parG\Msg' }  
 {\pp\ne\play{\asCom} }
\\[4ex]
\NamedRule{{\rn{Inside-In}}}
 {\G\parG \Msg\stackred\asCom\G'\parG \Msg'}
 {\agtI \pp\q i I \la \G\parG \mq\q{\la}\pp\cdot\Msg \stackred\asCom
 \agtI \pp\q i {I} \la {\G'}\parG  {\mq\q{\la}\pp\cdot\Msg'} 
}  
 {
 \pp\ne\play{\asCom}  
 }
\end{array}
\end{math} 
\caption{LTS for type configurations. }\label{fig:ltsgtAs}
\end{figure}

In order to show that  a type configuration  
does represent a correct and complete
description of the overall behaviour of a session (see Subject Reduction and Session Fidelity theorems), we equip  type  configurations with an LTS, as formally defined in Figure \ref{fig:ltsgtAs}.
Actually we are interested in reducing only type configurations $\G\parG \Msg$ such that $\tyng\pas{\G}{\Nt\parN\Msg}$ for some $\pas$ and $\Nt$. This justifies the shapes of message queues in Rules $\rn{Inside-Out}$ and $\rn{Inside-In}$,  which mimic the message queues in Rules $\rn{Out}$ and $\rn{In}$, see Figure~\ref{fig:cntr}.  The condition $\pp\ne\play{\asCom}$ in these rules ensures that  $\asCom$ is independent of the enclosing communication.

\FloatBarrier

 \section{Properties of Typable Sessions}\label{pr}

We begin with a few technical lemmas enabling to prove Subject reduction and Session Fidelity,
that is completeness and correctness, respectively, of type configurations with respect to sessions
(by taking into account participants in $\pas$ only).
These in turn will enable us to prove partial communication properties for typable sessions.

A first lemma immediately follows by cases on the typing axioms/rules.

\begin{lemma}\label{aa} 
If $\tyng\pas{\G}{\Nt\parG\Msg}$, then 
$(\Plays{\Nt}\setminus\Plays{\G})\cap\pas=\emptyset$ and $\G\parG\Msg$ is $\pas$-sound.
\end{lemma}

The following lemma  allows to  get rid of histories  in particular derivations.  
It  states that,  if a judgement occurs in a proof whose conclusion is without history, then the judgement itself holds without history. Moreover,  if the premises of Rules $\rn{Out}$ and $\rn{In}$  hold without histories, also the conclusion holds without history.

\begin{lemma}\label{key} $\;$
\begin{enumerate}
\item\label{key2} 
If $\Hset\tyng\pas{\G}{\Nt\parN\Msg}$ occurs in the proof of  $\tyng\pas{\G'}{\Nt'\parN\Msg'}$, then 
$\tyng\pas{\G}{\Nt\parN\Msg}$.
\item\label{key3} 
If $\tyng\pas{\G_i}{\pP\pp{\PP_i}\parN\Nt\parN\Msg\cdot\mq\pp{\la_i}{\q}}$ for all $i\in I$,  then $\tyng\pas{\agtO{\pp}{\q}{i}{I}{\la}{\G}}{\pP\pp{\oup{\q}{i}{I}{\la}{\PP}}\parN\Nt\parN\Msg}$.
\item\label{key4} 
If $\tyng\pas{\G}{\pP\pp{\PP_h}\parN\Nt\parN\Msg}$ and $h\in I$,  then $\tyng\pas{\agtI{\pp}{\q}{i}{I}{\la_h}{\G}}{\pP\pp{\inp{\q}{i}{I}{\la}{\PP}}\parN\Nt\parN\mq\q{\la_h}{\pp}\cdot\Msg}$.
\end{enumerate}
\end{lemma}
\begin{proof}
\ref{key2}.
By induction on the distance $d$ between $\Hset\tyng\pas{\G}{\Nt\parN\Msg}$ and
$\tyng\pas{\G'}{\Nt'\parN\Msg'}$ in the derivation of  $\tyng\pas{\G'}{\Nt'\parN\Msg'}$.
The case $d=0$ is trivial.

{\em Case $d=1$}.
Then $\Hset\tyng\pas{\G}{\Nt\parN\Msg}$ is a premise of a rule whose conclusion is 
$\tyng\pas{\G'}{\Nt'\parN\Msg'}$,  which implies $\Hset=\pair{\Nt'\parN\Msg'}{\G'}$. 
We can now build a derivation of $\tyng\pas{\G}{\Nt\parN\Msg}$ out of the derivation of  
$\pair{\Nt'\parN\Msg'}{\G'}\tyng\pas{\G}{\Nt\parN\Msg}$, as follows. 
First we erase everywhere $\pair{\Nt'\parN\Msg'}{\G'}$ from the histories present in the derivation.
This operation does not affect the correctness of the applicability conditions of
Axiom \rn{End} and Rules \rn{In} and \rn{Out}.
Axiom \rn{Cycle}, instead, is affected by such an erasing only in case  $\pair{\Nt'\parN\Msg'}{\G'}$
is the triple used in the axiom, namely $\Hset',\pair{\Nt'\parN\Msg'}{\G'}\tyng\pas{\G'}{\Nt'\parN\Msg'}$  is the axiom conclusion.
In such a case,  we replace this application of  Axiom  $\rn{Cycle}$ by  a proof of 
$\Hset'\tyng\pas{\G'}{\Nt'\parN\Msg'}$ built out of  the derivation $\mathcal{D}$ of  
$\tyng\pas{\G'}{\Nt'\parN\Msg'}$ in the following way.\\
Let us consider the premises of the last rule in the derivation $\mathcal{D}$. 
 For the premises which are axioms there is nothing to do. For the other premises we need to modify the derivation as follows.\\ Let 
$\pair{\Nt'\parN\Msg'}{\G'}\tyng\pas{\widehat\G}{\widehat\Nt\parN\widehat\Msg}$
 be   obtained as conclusion of
either Rule \rn{In} or Rule \rn{Out} with premises having
$\pair{\Nt'\parN\Msg'}{\G'}, \pair{\widehat\Nt\parN\widehat\Msg}{\widehat\G}$
as histories.
$\mathcal{D}$ has hence the form

\ClineL{
\prooftree
\quad \cdots\quad
\prooftree
\quad\cdots \quad \pair{\Nt'\parN\Msg'}{\G'}, \pair{\widehat\Nt\parN\widehat\Msg}{\widehat\G} \tyng\pas{\_}{\_} \quad \dots\quad
\justifies
\pair{\Nt'\parN\Msg'}{\G'}\tyng\pas{\widehat\G}{\widehat\Nt\parN\widehat\Msg}
\using \rn{In}/\rn{Out}
\endprooftree
\quad\cdots\quad
\justifies
\tyng\pas{\G'}{\Nt'\parN\Msg'}
\using \rn{In}/\rn{Out}
\endprooftree
}

\noindent
Notice that this implies that $\pair{\widehat\Nt\parN\widehat\Msg}{\widehat\G}\in \Hset'$.
We can hence transform the above derivation in a derivation of $\Hset'\tyng\pas{\G'}{\Nt'\parN\Msg'}$   as follows:

\ClineL{
\prooftree
\quad \cdots\quad
\prooftree
\justifies
\Hset',\pair{\Nt'\parN\Msg'}{\G'}\tyng\pas{\widehat\G}{\widehat\Nt\parN\widehat\Msg}
\using \rn{Cycle}
\endprooftree
\quad\cdots\quad
\justifies
\Hset'\tyng\pas{\G'}{\Nt'\parN\Msg'}
\using \rn{In}/\rn{Out}
\endprooftree
}

{\em Case $d>1$}.
Let  $\pair{\Nt'\parN\Msg'}{\G'}\tyng\pas{\widehat\G}{\widehat\Nt\parN\widehat\Msg}$  be an  arbitrary  premise of a rule whose conclusion is $\tyng\pas{\G'}{\Nt'\parN\Msg'}$.
By the construction described in the base case, we can get a derivation $\mathcal{D}$ for 
$\tyng\pas{\widehat\G}{\widehat\Nt\parN\widehat\Msg}$
containing a subderivation for 
$\Hset\setminus\set{\pair{\Nt'\parN\Msg'}{\G'}}\tyng\pas{\G}{\Nt\parN\Msg}$.
In $\mathcal{D}$ the distance between  
$\tyng\pas{\widehat\G}{\widehat\Nt\parN\widehat\Msg}$
and
$\Hset\setminus\set{\pair{\Nt'\parN\Msg'}{\G'}}\tyng\pas{\G}{\Nt\parN\Msg}$ is
$d-1$. So, by the induction hypothesis, we  conclude 
$\tyng\pas{\G}{\Nt\parN\Msg}$. 

\ref{key3}. Let $\G=\agtO{\pp}{\q}{i}{I}{\la}{\G}$ and $\Nt'=\pP\pp{\oup{\q}{i}{I}{\la}{\PP}}\parN\Nt$.  If a statement $\Hset\tyng\pas{\G}{\Nt'\parN\Msg'}$ does not occur in the derivation of 
$\tyng\pas{\G_i}{\pP\pp{\PP_i}\parN\Nt\parN\Msg\cdot\mq\pp{\la_i}{\q}}$,  then 
we can simply add  $\pair{\Nt'\parN\Msg}{\G}$  to the
histories of the derivation, so getting a still correct derivation, and then apply Rule $\rn{Out}$.
Otherwise this statement must also be the conclusion of an application of Rule $\rn{Out}$ with premises ${\Hset,\pair{\Nt'\parN\Msg'}\G}\tyng\pas{\G_i}{\pP\pp{\PP_i}\parN\Nt\parN\Msg'\cdot\mq\pp{\la_i}{\q}}$ for all $i\in I$. 
This implies $\Msg'\equiv\Msg$.
 Since $\Hset\tyng\pas{\G}{\Nt'\parN\Msg}$ occurs in a derivation of  $\tyng\pas{\G_i}{\pP\pp{\PP_i}\parN\Nt\parN\Msg\cdot\mq\pp{\la_i}{\q}}$,  by Point~(\ref{key2}) we conclude $\tyng\pas{\G}{\Nt'\parN\Msg}$.

\ref{key4}. Let $\G'=\agtI{\pp}{\q}{i}{I}{\la_h}{\G}$ and $\Nt'=\pP\pp{\inp{\q}{i}{I}{\la}{\PP}}\parN\Nt$ and $\Msg'\equiv\mq\q{\la_h}\pp\cdot\Msg$. 
If the derivation of 
$\tyng\pas{\G}{\pP\pp{\PP_h}\parN\Nt\parN\Msg}$ does not contain a statement $\Hset\tyng\pas{\G'}{\Nt'\parN\Msg''}$,  then
we can simply add  $\pair{\Nt'\parN\Msg}{\G'}$  to the
histories of the derivation, so getting a still correct derivation, and then apply Rule $\rn{In}$. 
Otherwise this statement must also be the conclusion of an application of Rule $\rn{In}$ with premise  ${\Hset,\pair{\Nt'\parN\Msg''}{\G'}}\tyng\pas{\G}{\pP\pp{\PP_h}\parN\Nt\parN\Msg}$.  
This implies $\Msg''\equiv\Msg'$.
 Since $\Hset\tyng\pas{\G'}{\Nt'\parN\Msg'}$ occurs in a derivation of $\tyng\pas{\G}{\pP\pp{\PP_h}\parN\Nt\parN\Msg}$, by Point~(\ref{key2}) we conclude $\tyng\pas{\G'}{\Nt'\parN\Msg'}$.
\end{proof}

 We can now show that the reductions of type configurations are matched by the reductions of the sessions. 

\begin{theorem}[Session Fidelity]\label{thm:sf}
If $\tyng\pas{\G}{\Nt\parN\Msg}$ and 
$\G\parG\Msg\stackred\co\G'\parN\Msg'$, then $\Nt\parN\Msg\stackred{\co}\Nt'\parN\Msg'$ and  $\tyng\pas{\G'}{\Nt'\parN\Msg'}$. 
\end{theorem}
\begin{proof} The proof is by cases on the last applied axiom/rule in the derivation of $\tyng\pas\G{\Nt\parN\Msg}$.

 {\em Axiom} \rn{$\End$}. Impossible since $\G=\End$  and 
 $\End\parG\Msg \not\rightarrow$.

{\em Axiom} \rn{Cycle}.  Impossible since the history cannot be empty.

{\em Rule} \rn{Out}.
In such a case $\G=\agtO{\pp}{\q}{i}{I}{\la}{\G}$ and
$\Nt = \pP\pp{\PP}\parN\widehat\Nt$
and $\PP=\oup{\q}{i}{I}{\la}{\PP}$ and  
\begin{equation}\label{eq1}
 \pair{\pP\pp{\PP}\parN\widehat\Nt\parN\Msg}\G  \tyng\pas{\G_i}{\pP\pp{\PP_i}\parN\widehat\Nt\parN\Msg\cdot\mq\pp{\la_i}{\q}}\quad \forall  i  \in I 
\end{equation}
We proceed by induction on  the 
 height $t$ of the derivation of $\G\parG\Msg \stackred{\co} \G'\parG\Msg'$.\\
\underline{\em Case $t=1$.} Then $\G\parG\Msg \stackred{\co} \G'\parG\Msg'$ is necessarily obtained
by Axiom \rn{Top-Out}, that is  $\co = \CommAs\pp{\la_h}{\q}$  for some  $h\in I$, and

\ClineL{
\NamedRule{\rn{Top-Out}}
 { } 
 {\agtO{\pp}{\q}i I{\la}{\G}\parG\Msg 
 \stackred{\CommAs\pp{\la_h}{\q}}
 \G_h\parG\addMsg\Msg{\mq\pp{\la_h}{\q}}}  {h\in I}
}

\noindent
By Rule \rn{Send} we have that

\ClineL{
 {\pP\pp{\PP}\parN\widehat\Nt\parN\Msg}
\stackred{\CommAs\pp{\la_h}{\q}}
\pP\pp{\PP_h}\parN\widehat\Nt\parN\Msg\cdot\mq\pp{\la_h}{\q} 
} 

\noindent
Now, by (\ref{eq1}) we have
$\pair{\pP\pp{\PP}\parN\Nt\parN\Msg}\G  \tyng\pas{\G_h}{\pP\pp{\PP_h}\parN\Nt\parN\Msg\cdot\mq\pp{\la_h}{\q}}$, and 
by Lemma \ref{key}(\ref{key2})
we get the rest of the thesis, namely 
$\tyng\pas{\G_h}{\pP\pp{\PP_h}\parN\widehat\Nt\parN\Msg\cdot\mq\pp{\la_h}{\q}}$. 

\noindent
 \underline{\em Case $t>1$.} Then $\G\parG\Msg \stackred{\co} \G'\parG\Msg'$ is necessarily obtained by  Rule  \rn{Inside-Out}, that is 
 
 \ClineL{
 \NamedRule{{\rn{Inside-Out}}}
 {\G_i\parG \Msg\cdot\mq\pp{\la_i}{\q} \stackred\asCom\G'_i \parG \Msg'\cdot\mq\pp{\la_i}{\q} \quad \forall i \in I} 
 {\agtO{\pp}{\q}i I{\la}{\G}\parG\Msg \stackred \asCom\agtO{\pp}{\q}i I{\la}{\G'}\parG\Msg'}  
 {\pp\ne\play{\asCom} }
 }
 
 \noindent
 From (\ref{eq1}) and Lemma \ref{key}(\ref{key2})   
 we can infer that 
 
 \ClineL{
  \tyng\pas{\G_i}{\pP\pp{\PP_i}\parN\widehat\Nt\parN\Msg\cdot\mq\pp{\la_i}{\q}}
  \quad \forall  i  \in I 
 }
 
\noindent
We can now recur to the induction hypothesis, getting

\ClineL{
\pP\pp{\PP_i}\parN\widehat\Nt\parN\Msg\cdot\mq\pp{\la_i}{\q} 
\stackred{\co}
\pP\pp{\PP_i}\parN\widehat\Nt'\parN\Msg'\cdot\mq\pp{\la_i}{\q}
\qquad \forall  i  \in I 
}

\noindent
and

\ClineL{
  \tyng\pas{\G'_i}{ \pP\pp{\PP_i}\parN\widehat\Nt'\parN\Msg'\cdot\mq\pp{\la_i}{\q} } \qquad \forall  i  \in I 
 }
 
\noindent
Notice that the condition $\pp\ne\play{\asCom}$ ensures that, for each $i\in I$,
the transition  does not modify the  process  of participant $\pp$.
Moreover, the  transition  does not depend on the messages
  $\mq\pp{\la_i}{\q} $,  since these messages are at the end of
  the queue both before and after the transitions. 
So, we can infer that 

  \ClineL{
   \pP\pp{\PP}\parN\widehat\Nt\parN\Msg
  \stackred{\co}
  \pP\pp{\PP}\parN\widehat\Nt'\parN\Msg' 
  }
  
  \noindent
Lemma \ref{key}(\ref{key3}) 
applied to 

  \ClineL{
   \tyng\pas{\G'_i}{ \pP\pp{\PP_i}\parN\widehat\Nt'\parN\Msg'\cdot\mq\pp{\la_i}{\q}} \text{ for all }i\in I
  }
  
  \noindent
   gives
  $
  \tyng\pas{\G'}{\pP\pp{\PP}\parN\widehat\Nt'\parN\Msg'}$.

{\em Rule} \rn{In}.
 In such a case $\G=\agtI{\pp}{\q}i I{\la_h}{\G'}$ and
$\Nt= \pP\pp{\PP}\parN\widehat\Nt$ and $\Msg\equiv\mq\q{\la_h}\pp\cdot \widehat\Msg$, with $h\in I$ and  $\PP=\inp{\q}{i}{I}{\la}{\PP}$ and 

\begin{equation}\label{eq2}
\pair{\pP\pp{\PP}\parN\widehat\Nt\parN\Msg}\G  \tyng\pas{\G'}{\pP\pp{\PP_{h}}\parN\widehat\Nt\parN\widehat\Msg}
\end{equation}

\noindent
We proceed by induction on  the 
 height $t$ of the derivation of $\G\parG\Msg \stackred{\co} \G'\parG\Msg'$.\\
\underline{\em Case $t=1$.} Then $\G\parG\Msg \stackred{\co} \G'\parG\Msg'$ is necessarily obtained
by Axiom \rn{Top-In}, that is $\co = \CommAsI{\pp}{\la_h}\q$ and 

\ClineL{
\NamedRule{{\rn{Top-In}}}
 { } 
 {\agtI \pp\q i I{\la_h} {\G'}\parG\addMsg{\mq{\q}{\la_h}\pp}\widehat\Msg
 \stackred{\CommAsI{\pp}{\la_h}\q}
 \G'\parG\widehat\Msg}
 {}
}

\noindent
By  Rule  \rn{Rcv} we have that 

\ClineL{
{\pP\pp{\PP}\parN\widehat\Nt\parN\addMsg{\mq{\q}{\la_h}\pp}\widehat\Msg}
\stackred{\CommAsI{\pp}{\la_h}\q}
{\pP\pp{\PP_h}\parN\widehat\Nt\parN\widehat\Msg}
}

\noindent
We can now get the rest of the thesis since (\ref{eq2}) and Lemma \ref{key}(\ref{key2}) 
imply 

\ClineL{
 \tyng\pas{\G'}{\pP\pp{\PP_{h}}\parN\widehat\Nt\parN\widehat\Msg}
}

\noindent
\underline{\em Case $t>1$.} 
Then $\G\parG\Msg \stackred{\co} \G'\parG\Msg'$ is necessarily obtained by Rule \rn{Inside-In}, that is

\ClineL{
\NamedRule{{\rn{Inside-In}}}
 {\G'\parG \widehat\Msg\stackred\asCom\G''\parG \widehat\Msg'}
 { \agtI \pp\q i I {\la_h} {\G'}\parG \addMsg{\mq{\q}{\la_h}\pp}{\widehat\Msg} \stackred\asCom
 \agtI \pp\q i {I} {\la_h} {\G''}\parG  {\addMsg{\mq{\q}{\la_h}\pp}{\widehat\Msg'}} 
}  
 {\begin{array}{c} 
 \pp\ne\play{\asCom} 
 \end{array}}
}

\noindent
Now, (\ref{eq2}) and Lemma \ref{key}(\ref{key2})  imply 

\ClineL{
 \tyng\pas{\G'}{\pP\pp{\PP_{h}}\parN\widehat\Nt\parN\widehat\Msg} 
}

\noindent
We can hence recur to the induction hypothesis, getting that

\ClineL{
\pP\pp{\PP_{h}}\parN\widehat\Nt\parN\widehat\Msg
\stackred\asCom
\pP\pp{\PP_{h}}\parN\widehat\Nt'\parN\widehat\Msg'   
}

\noindent
and 

\ClineL{
\tyng\pas{\G''}{\pP\pp{\PP_{h}}\parN\widehat\Nt'\parN\widehat\Msg'} 
}

\noindent
since the side condition $\pp\ne\play{\asCom}$ of Rule \rn{Inside-In} implies that the process of participant $\pp$ is unchanged.
Lemma \ref{key}(\ref{key4}) applied to 

\ClineL{
\tyng\pas{\G''}{\pP\pp{\PP_{h}}\parN\widehat\Nt'\parN\widehat\Msg'} 
}

\noindent
gives the rest of the thesis, namely

\vspace{3mm}

 $\qquad\qquad\qquad\qquad\qquad\tyng\pas{ \agtI \pp\q i {I} {\la_h} \G''}{\pP\pp{\PP}\parN\widehat\Nt'\parN\mq{\q}{\la_h}\pp\cdot\widehat\Msg'}$  \end{proof}

 The proof of Subject Reduction requires some lemmas which are typical of our partial typing. 
The first lemma deals with participants which have active processes but are not players of global types. The second lemma deals with messages whose receivers are not  players of global types. 
The last lemma states that a player of the network whose  lock-freedom  
must be ensured is always a player of the global type.

\begin{lemma}\label{a} 
If  $\tyng\pas\G{\pP\pp\PP\parN\Nt\parN\Msg}$, $\PP\neq\inact$ and $\pp\not\in\Plays{\G}$, then $\tyng\pas\G{\pP\pp{\PP'}\parN\Nt\parN\Msg}$  for  any  arbitrary $\PP'$ such that $\pp\not\in\plays{\PP'}$.
\end{lemma}
\begin{proof} 
If  $\pp\not\in\Plays{\G}$, then the process $\PP$ can never be involved in any occurrence of Rules  $\rn{Out}$ or $\rn{In}$.  This implies that $\pP\pp\PP$ must occur only in axioms. It is hence enough to replace $\PP$ by $\PP'$ in those axioms and 
modify the histories present in the derivation accordingly.
\end{proof}

\begin{lemma}\label{b} 
If  $\tyng\pas\G{\Nt\parN\addMsg{\mq{\q}{\la}\pp}\Msg}$ and $\pp\not\in\Plays{\G}$, then  $\tyng\pas\G{\Nt\parN\Msg}$.
\end{lemma}
\begin{proof} Rule \rn{Out} does not add messages to the queue. 
If  $\pp\not\in\Plays{\G}$, then $\mq{\q}{\la}\pp$ cannot be added by  
Rule \rn{In}.  
Then $\mq{\q}{\la}\pp$ is present in all the queues of the judgements in the derivation. We remark that the removal of a message from a queue cannot alter the truth value of the $\pas$-soundness
condition, which is required for the applicability of an axiom or a rule. It is hence possible  to remove $\mq{\q}{\la}\pp$ from the queues in the axioms and modify the queues present in the derivation accordingly. 
\end{proof}

\begin{lemma}\label{aux}
If $\tyng\pas{\G}{\Nt\parN\Msg}$ and $\pp\in(\Plays{\Nt}\cap\pas)$,  then $\pp\in\Plays\G$.
\end{lemma}
\begin{proof}
If $\pp\in\pas$, then an output with sender $\pp$ can only be typed by Rule $\rn{Out}$ and an input with receiver $\pp$ together with a message with receiver $\pp$ can only be typed by Rule $\rn{In}$. 
\end{proof}

 Subject Reduction ensures that a transition of a session is mimicked by a transition of the corresponding type configuration only if the player of the transition is a player of the global type.

\begin{theorem}[Subject Reduction]\label{thm:sr}
Let $\tyng\pas{\G}{\Nt\parN\Msg}$ and $\Nt\parN\Msg\stackred\beta\Nt'\parN\Msg'$. 
If $\play\beta\in\Plays\G$,  
then $\G\parG\Msg\stackred\beta\G'\parG\Msg'$ and $ \tyng\pas{\G'}{\Nt'\parN\Msg'}$.  Otherwise $ \tyng\pas{\G}{\Nt'\parN\Msg'}$.
\end{theorem}
\begin{proof} The proof is by cases on the reduction rules.

{\bf Rule} $\rn{Send}$. In this case 

\ClineL{
\text{$\Nt\equiv\pP{\pp}{\oup\q{i}{I}{\la}{\PP}}\parN\Nt_0$, \quad $\Msg'\equiv\addMsg{\Msg}{\mq{\pp}{\la_h}\q}$, \quad $\beta=\CommAs\pp{\la_h}{\q}$, \quad $\Nt'\equiv\pP{\pp}{\PP_h}\parN\Nt_0$  \quad where  $h \in I$.}
} 

\noindent
By definition of network $\pp\not\in\plays{\oup\q{i}{I}{\la}{\PP}}$, which implies $\pp\not\in\plays{\PP_h}$. 
If $\pp\not\in\Plays\G$, Lemma \ref{a} implies  
$ \tyng\pas{\G}{\Nt'\parN\Msg'}$.  Otherwise the proof proceeds by cases on the last typing axiom/rule used in the derivation for $\tyng\pas{\G}{\Nt\parN\Msg}$ and by induction on $d=\weight(\G,\pp)$.\\
{\em Axiom $\rn{End}$.} Since it cannot be $\play\beta\in\Plays\G$, the implication is vacuously satisfied.\\
{\em  Axiom $\rn{Cycle}$.}  Impossible since the history cannot be empty.\\ 
{\em Rule $\rn{Out}$ and $d=1$.} In this case $\G=\agtO{\pp}{\q}{i}{I}{\la}{\G}$. We get $\G\parN\Msg\stackred{\CommAs\pp{\la_h}{\q}}\G_h\parN\Msg'$ by Axiom $\rn{Top-Out}$. Lemma \ref{key}(\ref{key2}) implies 

\ClineL{ \tyng\pas{\G_h}{\pP\pp{\PP_h}\parN\Nt_0\parN\addMsg{\Msg}{\mq{\pp}{\la_h}\q}}}

\noindent
{\em Rule $\rn{Out}$ and  $d>1$.} In this case $\G=\agtO{\pr}{\ps}{j}{J}{\la'}{\G}$ with $\pr\neq\pp$ and $\Nt_0\equiv\pP{\pr}{\oup\ps{j}{J}{\la'}{\R}}\parN\Nt_1$. Lemma \ref{key}(\ref{key2}) implies 

\ClineL{
\text{
$\tyng\pas{\G_j}{\pP{\pp}{\oup\q{i}{I}{\la}{\PP}}\parN\pP\pr{\R_j}\parN\Nt_1\parN\addMsg{\Msg}{\mq{\pr}{\la'_j}\ps}}$ \quad for all $j\in J$ }
}

\noindent
We hence get, by Rule \rn{Send}, for all $j\in J$,

\ClineL{
\pP{\pp}{\oup\q{i}{I}{\la}{\PP}}\parN\pP\pr{\R_j}\parN\Nt_1\parN\addMsg{\Msg}{\mq{\pr}{\la'_j}\ps}\stackred{\CommAs\pp{\la_h}{\q}}\pP{\pp}{\PP_h}\parN\pP\pr{\R_j}\parN\Nt_1\parN\addMsg{\Msg}{\mq{\pr}{\la'_j}\ps}\cdot\mq{\pp}{\la_h}\q
}

\noindent
Since $\weight(\G_j,\pp)<d$, induction implies 
 $\G_j\parN\addMsg{\Msg}{\mq{\pr}{\la'_j}\ps}\stackred{\CommAs\pp{\la_h}{\q}}\G'_j\parN\addMsg{\Msg}{\mq{\pr}{\la'_j}\ps}\cdot\mq{\pp}{\la_h}\q$ 
 and 
 
 \ClineL{
\text{ $\tyng\pas{\G'_j}{\pP{\pp}{\PP_h}\parN\pP\pr{\R_j}\parN\Nt_1\parN\addMsg{\Msg}{\mq{\pr}{\la'_j}\ps}\cdot\mq{\pp}{\la_h}\q}$
 \quad for all $j\in J$}
 } 
 
 \noindent
 Let $\G'=\agtO{\pr}{\ps}{j}{J}{\la'}{\G'}$ and  $\Msg' = \Msg\cdot\mq{\pp}{\la_h}\q$. 
  Since the messages $\mq{\pr}{\la'_j}\ps$ and $\mq{\pp}{\la_h}\q$ commute, being $\pr\neq\pp$, we can derive $\G\parN\Msg\stackred{\CommAs\pp{\la_h}{\q}}\G'\parN\Msg'$ using  Rule $\rn{Inside-Out}$.  Lastly, $\tyng\pas{\G'}{\Nt'\parN\Msg'}$  by Lemma \ref{key}(\ref{key3}).\\
 {\em Rule $\rn{In}$ and $d=1$.} Impossible.\\
 {\em Rule $\rn{In}$ and $d>1$.}  In this case 
 
 \ClineL{
 \text{$\G=\agtI{\pr}{\ps}{j}{J}{\la'_k}{\G''}$ with $\pr\neq\pp$ and $\Nt_0\equiv\pP{\pr}{\inp\ps{j}{J}{\la'}{\R}}\parN\Nt_1$ 
 and $\Msg\equiv\mq\ps{\la'_k}\pr\cdot\Msg_0$ with $k\in J$.}
 }
 
 \noindent
Moreover, $\tyng\pas{\G''}{\pP{\pp}{\oup\q{i}{I}{\la}{\PP}}\parN\pP\pr{\R_k}\parN\Nt_1\parN\Msg}$  by Lemma \ref{key}(\ref{key2}). 
We get 

 \ClineL{
 \pP{\pp}{\oup\q{i}{I}{\la}{\PP}}\parN\pP\pr{\R_k}\parN\Nt_1\parN\Msg\stackred{\CommAs\pp{\la_h}{\q}}\pP{\pp}{\PP_h}\parN\parN\pP\pr{\R_k}\parN\Nt_1\parN\Msg\cdot\mq{\pp}{\la_h}\q
 }
 
 \noindent
  Since $\weight(\G'',\pp)<d$, induction implies
  
  \ClineL{
 \G''\parN\Msg\stackred{\CommAs\pp{\la_h}{\q}}\G'''\parN\Msg\cdot\mq{\pp}{\la_h}\q
 \quad\text{and} \quad
 \tyng\pas{\G'''}{\pP{\pp}{\PP_h}\parN\pP\pr{\R_k}\parN\Nt_1\parN\Msg\cdot\mq{\pp}{\la_h}\q}
 }
 
 \noindent
 Let $\G'=\agtI{\pr}{\ps}{j}{J}{\la'_k}{\G'''}$.
 Being $\pr\neq\pp$ we can derive $\G\parN\Msg\stackred{\CommAs\pp{\la_h}{\q}}\G'\parN\Msg'$ using  Rule $\rn{Inside-In}$.  Lastly, $\tyng\pas{\G'}{\Nt'\parN\Msg'}$  by Lemma \ref{key}(\ref{key4}).
 
{\bf Rule} $\rn{Rcv}$. In this case

\ClineL{
\text{$\Nt\equiv\pP{\pp}{\inp\q{i}{I}{\la}{\PP}}\parN\Nt_0$,\quad $\Msg\equiv\addMsg{\mq{\q}{\la_h}\pp}{\Msg'}$,\quad $\beta=\CommAsI\pp{\la_h}{\q}$,\quad $\Nt'\equiv\pP{\pp}{\PP_h}\parN\Nt_0$ \quad where  $h \in I$.}
}

\noindent
By definition of network $\pp\not\in\plays{\inp\q{i}{I}{\la}{\PP}}$, which implies $\pp\not\in\plays{\PP_h}$. 
  If $\pp\not\in\plays\G$  Lemmas \ref{a} and \ref{b} imply  $ \tyng\pas{\G}{\Nt'\parN\Msg'}$. Otherwise the proof proceeds by cases on the last axiom/rule used in the derivation for $\tyng\pas{\G}{\Nt\parN\Msg}$ and by induction on  $d=\weight(\G,\pp)$.\\
{\em Axiom $\rn{End}$}. Since it cannot be $\play\beta\in\plays\G$, the implication is vacuously satisfied.\\
{\em Axiom $\rn{Cycle}$.}  Impossible since the history cannot be empty. \\
{\em Rule $\rn{Out}$ and $d=1$.} Impossible. \\
{\em Rule $\rn{Out}$ and  $d>1$.}
 In this case $\G=\agtO{\pr}{\ps}{j}{J}{\la'}{\G}$ with $\pr\neq\pp$ and $\Nt_0\equiv\pP{\pr}{\oup\ps{j}{J}{\la'}{\R}}\parN\Nt_1$ and $\tyng\pas{\G_j}{\pP{\pp}{\inp\q{i}{I}{\la}{\PP}}\parN\pP\pr{\R_j}\parN\Nt_1\parN\addMsg{\Msg}{\mq{\pr}{\la'_j}\ps}}$ for all $j\in J$ by Lemma \ref{key}(\ref{key2}). 
We get, for all $j\in J$,

\ClineL{
\pP{\pp}{\inp\q{i}{I}{\la}{\PP}}\parN\pP\pr{\R_j}\parN\Nt_1\parN\addMsg{\Msg}{\mq{\pr}{\la'_j}\ps}
\stackred{\CommAsI\pp{\la_h}{\q}}
\pP{\pp}{\PP_h}\parN\pP\pr{\R_j}\parN\Nt_1\parN\addMsg{\Msg'}{\mq{\pr}{\la'_j}\ps}
}

\noindent
 Since $\weight(\G_j,\pp)<d$, induction implies,  for all $j\in J$,
 
 \ClineL{
 \G_j\parN\addMsg{\Msg}{\mq{\pr}{\la'_j}\ps}\stackred{\CommAsI\pp{\la_h}{\q}}\G'_j\parN\addMsg{\Msg'}{\mq{\pr}{\la'_j}\ps}
 \quad\text{and} \quad
 \tyng\pas{\G'_j}{\pP{\pp}{\PP_h}\parN\pP\pr{\R_j}\parN\Nt_1\parN\addMsg{\Msg'}{\mq{\pr}{\la'_j}\ps}}
 }
 
 \noindent
 Let 
 $\G'=\agtO{\pr}{\ps}{j}{J}{\la'}{\G'}$.  Being $\pr\neq\pp$ 
  we can derive $\G\parN\Msg\stackred{\CommAsI\pp{\la_h}{\q}}\G'\parN\Msg'$ using  Rule $\rn{Inside-Out}$. Lastly, $\tyng\pas{\G'}{\Nt'\parN\Msg'}$ by Lemma \ref{key}(\ref{key3}).\\
{\em Rule $\rn{In}$ and $d=1$.}
 In this case $\G=\agtI{\pp}{\q}{i}{I}{\la}{\G'}$. Lemma \ref{key}(\ref{key2}) implies $\tyng\pas{\G'}{\pP\pp{\PP_h}\parN\Nt_0\parN\Msg'}$. 
We get $\G\parN\Msg\stackred{\CommAsI\pp{\la_h}{\q}}\G'\parN\Msg'$ by  Axiom  $\rn{Top-In}$. \\
{\em Rule $\rn{In}$ and $d>1$.}  
In this case 

\ClineL{
\text{$\G=\agtI{\pr}{\ps}{j}{J}{\la'_k}{\G''}$\quad $\Nt_0\equiv\pP{\pr}{\inp\ps{j}{J}{\la'}{\R}}\parN\Nt_1$,\quad $\Msg'\equiv\mq\ps{\la'_k}\pr\cdot\Msg_0$\quad  with $\pr\neq\pp$ and $k\in J$}
 }

\noindent
Lemma \ref{key}(\ref{key2}) implies $\tyng\pas{\G''}{\pP{\pp}{\inp\q{i}{I}{\la}{\PP}}\parN\pP\pr{\R_k}\parN\Nt_1\parN  \mq\q{\la_h}\pp\cdot \Msg_0}$. 
 We get 
 
 \ClineL{
 \pP{\pp}{\inp\q{i}{I}{\la}{\PP}}\parN\pP\pr{\R_k}\parN\Nt_1\parN\mq\q{\la_h}\pp\cdot\Msg_0
 \stackred{\CommAsI\pp{\la_h}{\q}}
 \pP{\pp}{\PP_h}\parN\pP\pr{\R_k}\parN\Nt_1\parN\Msg_0
 } 
 
 \noindent
 Since $\weight(\G'',\pp)<d$, induction implies
 
 \ClineL{
 \G''\parN\mq\q{\la_h}\pp\cdot\Msg_0\stackred{\CommAsI\pp{\la_h}{\q}}\G'''\parN\Msg_0
  \quad\text{and} \quad
  \tyng\pas{\G'''}{\pP{\pp}{\PP_h}\parN\pP\pr{\R_k}\parN\Nt_1\parN\Msg_0}
  } 
  
  \noindent
 Let $\G'=\agtI{\pr}{\ps}{j}{J}{\la_k'}{\G'''}$. Being $\pr\neq\pp$ we can derive $\G\parN\Msg\stackred{\CommAsI\pp{\la_h}{\q}}\G'\parN\Msg'$ using  Rule $\rn{Inside-In}$. Lastly, $\tyng\pas{\G'}{\Nt'\parN\Msg'}$ by Lemma \ref{key}(\ref{key4}).
\end{proof}
 
  We conclude  this section by showing the main properties of our type system: partial lock-freedom and partial orphan-message-freedom. 
 
\begin{theorem}[Partial Lock-freedom]\label{thm:aplf}
If $\tyng\pas{\G}{\Nt\parN\Msg}$,  then $\Nt\parN\Msg$ is $\pas$-lock free.  
\end{theorem}
\begin{proof} Let $\pp\in\pas$. If $\pp\not\in\Plays{\Nt}$, then $\Nt\parN\Msg$ is trivially $\pp$-lock free.
Otherwise $\pp\in(\Plays{\Nt}\cap\pas)$  gives
 $\pp\in\Plays\G$  by Lemma~\ref{aux}.  
 We first show by induction on $d=\weight(\G,\pp)$ that $\G\parN\Msg\stackred{\comseqA\cdot\co}$ with $\play\co=\pp$ for some $\comseqA$, $\co$.\\
If $d=1$, then either $\G=\agtO{\pp}{\q}{i}{I}{\la}{\G}$ or  $\G=\agtI{\pp}{\q}{i}{I}{\la}{\G'}$.  We get $\G\parN\Msg\stackred{\co}$ with $\play\co=\pp$ by either  Axiom  $\rn{Top-Out}$ or  Axiom  $\rn{Top-In}$.\\
If $d>1$, then $\G\parN\Msg\stackred{\co'}\G'\parN\Msg'$ for some $\co'$, $\G'$ and $\Msg'$ by  Axiom  $\rn{Top-Out}$ or  Axiom  $\rn{Top-In}$. The applicability of  Axiom  $\rn{Top-In}$ is ensured by the fact that $\tyng\pas{\G}{\Nt\parN\Msg}$ must be typed using Rule $\rn{In}$. Since $\weight(\G',\pp)< d$, by induction $\G'\parN\Msg'\stackred{\comseqA'\cdot\co}$ with $\play\co=\pp$ for some $\comseqA$, $\co$. We can take $\comseqA=\co'\cdot\comseqA'$. 

By Theorem \ref{thm:sf} $\G\parN\Msg\stackred{\comseqA\cdot\co}$ implies $\Nt\parN\Msg\stackred{\comseqA\cdot\co}$.
\end{proof}

\begin{theorem}[Partial Orphan-message-freedom]\label{thm:pomf}
If $\tyng\pas{\G}{\Nt\parN\Msg}$,  then $\Nt\parN\Msg$ is $\pas$-orphan-message free.  
\end{theorem}
\begin{proof} Let  $\Msg\equiv\mq\pp\la\q\cdot\widehat\Msg$  and $\set{\pp,\q}\subseteq\pas$.
 We first show that $\G\parG\Msg\stackred{\comseqA\cdot\q\pp?\la}$ by induction on $\wgs{\mq\pp\la\q}\G$. If $\wgs{\mq\pp\la\q}\G=0$ it is trivial. 
 Otherwise $\G\parG\Msg\stackred{\co}\G'\parG\Msg'$ by  Axiom  $\rn{Top-Out}$ or  Axiom  
 $\rn{Top-In}$ and $\wgs{\mq\pp\la\q}{\G'}<\wgs{\mq\pp\la\q}\G$. The applicability of  Axiom  $\rn{Top-In}$ is ensured by the fact that $\tyng\pas{\G}{\Nt\parN\Msg}$ must be typed using Rule $\rn{In}$. By induction $\G'\parG\Msg'\stackred{\comseqA'\cdot\q\pp?\la}$, so we can take $\comseqA=\co\cdot\comseqA'$. 

Applying Theorem~\ref{thm:sf} to $\G\parG\Msg\stackred{\comseqA\cdot\q\pp?\la}$ we conclude $\Nt\parN\Msg\stackred{\comseqA\cdot\q\pp?\la}$.
\end{proof}

It is worth noticing that  in case we were interested in 
$\pas$-lock-freedom only  we could simply take out  
the $\pas$-soundness conditions in the type system.

\begin{remark}[Saving $\pas$-soundness checks]{\em
One could avoid to  have the  $\pas$-soundness condition 
in Rules $\rn{In}$
and $\rn{Out}$ in case we impose $\Msg=\emptyset$ in  Axiom  $\rn{Cycle}$. 
In fact $(a)$ $\G\parN\emptyset$  is $\pas$-sound for any $\G$ and $\pas$ and
$(b)$ applications of Rules $\rn{In}$ and $\rn{Out}$ do preserve $\pas$-soundness.
Note that requiring $\G\parN\Msg$  to be $\pas$-sound  only in  Axioms  $\rn{Cycle}$
and $\rn{End}$ would not work.
A counterexample being the obvious derivation for \\
\Cline{\tyng\pas\G{\pP\pp\PP \parN \pP\q\Q \parN \mq{\q}{\lambda}{\pp}}
}
where
$\PP = \q!\lambda.\PP$,  $\Q = \pp?\lambda.\Q$  and   $\G = \pp\q!.\q\pp?\lambda.\G$.
}
\end{remark}

 \section{Conclusions}\label{conc}

Membership of a component to a concurrent/distributed system 
does not imply  that the component is 
equivalent in rights, capabilities and properties to the other components.
A system can often  viewed as being formed by 
different and  heterogeneous  
subsystems. 
Formal verification techniques and methods are usually devised  to ensure
properties of whole systems and they cannot always be scaled down or tailored to work on
specific subsystems.
This is obviously due to non trivial interactions between subsystems and the rest of system
components.
This issue has been addressed in \cite{BD23}, in the development/verification 
framework of  MPTS.  
The type assignment of~\cite{BD23},  guaranteeing good
communication properties, can be in fact tailored for specific subsets of participants, so disregarding
the behaviour of the rest of the participants. 
In the present paper we extend the investigation in \cite{BD23}
by  considering an asynchronous model of communication, which was instead
synchronous in \cite{BD23}. 
With respect to that paper we consider, besides $\pas$-lock-freedom
(absence of locks for participants in $\pas$),  also  $\pas$-orphan-message freedom.
The type assignment we devise is inspired by~\cite{CDG21,CDG22,DGD22} where,
unlike most choreographic formalisms, the 
asynchronicity of the communication model is explicitly reflected
at the level of global-behaviour descriptions, namely the global types in our case.


A MPST formalism dealing  with  properties holding for partial descriptions of systems
was defined in~\cite{HY17} and  further investigated  in~\cite{CDG18,CDGH19}.
In those papers, a notion of {\em connecting communications} enables us to consider 
some participants as optional, in particular the ones that are ``invited'' (via connecting inputs) to join some interactions.
Such a feature allows for a more natural description of typical communication protocols.
Connecting communications and our partial typing are sort of orthogonal.
An advantage of connecting communications over partial typing 
(where participants offering connecting communications should be ignored) is that only participants offering connecting inputs can be stuck. The disadvantage is that the typing rules are more demanding, so many interesting sessions can be partially typed but cannot be typed using connecting communications. 
We definitely deem worth investigating an extension of our formalism to deal with
participants offering connecting communications.

An algorithm enabling to infer all the global types for a given session
-- and handling, in particular, infinite expressions as sets of recursive equations -- has been devised in \cite{BD23}, working on a similar one in \cite{DGD22}.
We are confident that the approach of \cite{DGD22}, for what concerns the representation of infinite terms, can be also exploited in inference algorithms for our system.

The MPTS  formalism used in the present paper, 
unlike many MPST formalisms stemmed from ~\cite{Honda2016},
does not recur to projections.
Extending the standard projection operator to a relation between global types
and local behaviours with good partial properties would lead to a 
{\em top-down} development/verification formalism for partial properties,
 i.e. where local descriptions are  obtained  
 by projecting
previously developed global descriptions.

The properties verified by formalisms like the present one, as well as the ones in
~\cite{BD23,CDG21,CDG22,DGD22}, are strictly related to LTSs on  type configurations. 
Such LTSs are inductively defined. It is worth considering
coinductively defined LTSs, so that communication properties can be ensured for wider sets of sessions.

\bigskip

{\bf Acknowledgements} We are grateful to the anonymous referees for their comments and suggestions to improve the
readability of this paper.

\bibliographystyle{eptcs}
\bibliography{session}

\end{document}